%% file: paper.tex
\documentclass[10pt,twocolumn,twoside]{IEEEtran}   

\IEEEoverridecommandlockouts                       

\input{preamble/common}
\usepackage[capitalise]{cleveref}

\DeclareMathAlphabet{\mathcal}{OMS}{cmsy}{m}{n}

\title{\LARGE \bf
 Distributed Online Task Assignment via Inexact ADMM for unplanned online tasks and its Applications to Security}
\author{Ziqi Yang and Roberto Tron \IEEEmembership{Member, IEEE} 
\thanks{This project is supported by the National Science Foundation grant "CPS: Medium: Collaborative Research: Multiagent Physical Cognition and Control Synthesis Against Cyber Attacks" (Award number 1932162).}
\thanks{Ziqi Yang is with the Department of Systems Engineering,
Boston University, Boston, MA 02215 USA (e-mail: zy259@bu.edu).
}
\thanks{Roberto Tron is with the Faculty of Mechanical Engineering and Systems Engineering, Boston University, Boston, MA 02215 USA (e-mail:tron@bu.edu).}}

\begin{document}

\maketitle
\thispagestyle{empty}
\pagestyle{empty}


\begin{abstract}
    In multi-robot system (MRS) applications, efficient task assignment is essential not only for coordinating agents and ensuring mission success but also for maintaining overall system security. In this work, we first propose an optimization-based distributed task assignment algorithm that dynamically assigns mandatory security-critical tasks and optional tasks among teams. Leveraging an inexact Alternating Direction Method of Multipliers (ADMM)-based approach, we decompose the task assignment problem into separable and non-separable subproblems. The non-separable subproblems are transformed into an inexact ADMM update by projected gradient descent, which can be performed through several communication steps within the team. 
    
    In the second part of this paper, we formulate a comprehensive framework that enables MRS under plan-deviation attacks to handle online tasks without compromising security. The process begins with a security analysis that determines whether an online task can be executed securely by a robot and, if so, the required time and location for the robot to rejoin the team. Next, the proposed task assignment algorithm is used to allocate security-related tasks and verified online tasks. Finally, task fulfillment is managed using a Control Lyapunov Function (CLF)-based controller, while security enforcement is ensured through a Control Barrier Function (CBF)-based security filter. Through simulations, we demonstrate that the proposed framework allows MRS to effectively respond to unplanned online tasks while maintaining security guarantees.
\end{abstract}
\begin{IEEEkeywords}
    Multi-robot system, distributed task assignment, optimization, cyber-physical security
  \end{IEEEkeywords}

\section{Introduction}\label{sec:introduction}
Effective team division is a common and valuable approach in multi-robot system (MRS) applications, particularly for coordinating tasks across large, dynamic environments. By forming sub-teams, MRS can enhance coordination, scalability, and task efficiency, enabling better coverage, redundancy, and adaptability. In security scenarios, such as patrolling, surveillance, and co-observation\cite{pajares2015overview,wardega2019resilience,julian2012distributed}, the sub-team structure offers distinct advantages over individual robot deployments by allowing tasks to be distributed more effectively within and across teams. Additionally, the flexible nature of MRS expends the integration of security measures into non-security tasks by augmenting additional application-specific objectives (e.g., map exploration) with co-observation plans (see, e.g., our previous work in \cite{wardega2019resilience} for details). This approach not only improves overall task performance but also enables the efficient monitoring of sensitive areas, detection of anomalies, and timely response to threats. 

These security applications often require teams to follow pre-defined trajectories, visit key checkpoints, and collaborate or co-observe with other teams to detect intrusions or suspicious activity. In practice, however, there might be situations that cannot be planned for ahead of time (e.g., visiting unexpected targets of interest). These online events will unavoidably lead to conflicts with the security spatio-temporal requirements (e.g., reaching an unexpected target might lead to missing a checkpoint). A complete optimal replan in such cases might be unfeasible, due to high-dimensional computations or security considerations. To take advantage of the redundancy provided by MRS and sub-teams, we propose to handle these \emph{online tasks} through an online task assignment that optimizes the use of the additional robots. The task assignment algorithm ensures that some robots will always fulfill the original plan and maintain the security requirements so that others can be assigned to handle online tasks. This requires the system to optimally assign tasks according to predefined priorities (pre-planned tasks versus online tasks)~\cite{khamis2015multi}. 

Traditionally, multi-agent task assignment problems can be addressed using either distributed or centralized approaches~\cite{chakraa2023optimization}. Centralized approaches require individual agent's information to be communicated to a central agent or server, which generates a plan for the entire system. By taking into account the capabilities and limitations of all agents and tasks, the central planner is able to optimize the overall system performance \cite{prasad2020,coltin2010mobile,liu2012centralized,jin2003cooperative}. However, the centralized method sacrifices robustness and scalability, exposing the system to risks of failures at the central entity and limiting the range to where centralized communications are possible.

Decentralized approaches, on the other hand, rely on local communications between agents, and offer resilience to a single point of failure, as well as scalability in team size and mission range. Common methods include local search algorithm~\cite{lee2014ad}, bio-inspired methods~\cite{jevtic2011distributed}, consensus-based algorithms~\cite{alighanbari2005decentralized,dionne2007multi}, auction algorithms~\cite{quinton2023market,cao2012overview}, game theory-based algorithms~\cite{Xu2024}, just to name a few. Consensus usually relies on local communication to converge to a common value\cite{alighanbari2005decentralized,dionne2007multi}. However, these methods cannot guarantee a conflict-free solution. On the other hand, auction algorithms with conflict-free capability are not robust to distributed network topologies, which may not cope with multiple tasks~\cite{sariel2005real,Bai2023}. Some approaches provide improvements for the latter case, e.g., by running sequential auctions \cite{sariel2005real,sujit2007distributed}, or incorporating consensus methods \cite{choi2009consensus}. In any case, however, it is not clear that existing consensus- or auction-based methods can be applied to our application, which requires real-time adaptability and robustness to dynamic changes in order to maintain security.

In this paper, we present a distributed task assignment algorithm to manage tasks with different priorities. The task assignment is formulated as an optimization problem, where task priorities are incorporated as constraints. We first demonstrate that when the number of tasks matches the number of robots, the problem can be solved in a distributed manner using a projected gradient descent variation of the inexact Alternating Direction Method of Multipliers (ADMM). As shown in \cite{ma2016alternating,chang2014multi}, the use of the inexact gradient-based updates in each ADMM step significantly reduces complexity while still ensuring convergence. With this approach, we divide the task assignment problem into separable and non-separable ADMM blocks. The separable block can be solved individually by each robot; by applying a projected gradient descent method in the ADMM update, the non-separable subproblems can also be solved in a distributed manner through communication between connected robots. In addition, we introduce the concepts of \emph{shadow agents} and \emph{secondary trajectory tasks} to address scenarios where the number of tasks is different from the number of available agents.

We test our task assignment algorithm in a co-observation-secured application, building on previous work in the area~\cite{yangmulti,wardega2019resilience}. \emph{Co-observations} serve as an additional security layer for non-security tasks, which utilize the physical sensing capabilities of robots to perform mutual observations within the MRS system. This technique aims to mitigate the risks posted by \emph{physical masquerade attacks}, where compromised robots are taken over by attackers, masquerading as legitimate robots, to gain unauthorized access to forbidden areas. As shown in~\cite{yang2021multi}, adversarial robots can be detected even without knowing the attacker's model by incorporating proper \emph{reachability} constraints and co-observation schedule into the trajectory plan~\cite{yangmulti,wardega2019resilience}. To extend this security framework to handle unplanned online tasks, we propose an online control framework that integrates co-observation and reachability constraints while ensuring safe task execution. The framework consists of three key components: 
\begin{itemize}
    \item \emph{Regroup time calculation module} that evaluates whether an unplanned task can be safely executed while maintaining security guarantees;
    \item \emph{Distributed task assignment algorithm} that dynamically allocates robots to both trajectory and online tasks without compromising co-observation security;
    \item \emph{Online control scheme} that trajectory tracking and online task fulfillment are formulated as Control Lyapunov Function (CLF) constraints. Security conditions are transformed into Signal Temporal Logic (STL) requirements and enforced using a Control Barrier Function (CBF)-STL-based security filter, ensuring real-time adherence to both spatial and timing constraints. 
\end{itemize}
 This framework allows robots to respond to unplanned events while preserving security guarantees, as detailed in~\cref{fig:Flowchart}.

\begin{figure}
    \centering
    \includegraphics[width=0.9\linewidth, trim = 5.2cm 10.5cm 3.3cm 4cm, clip]{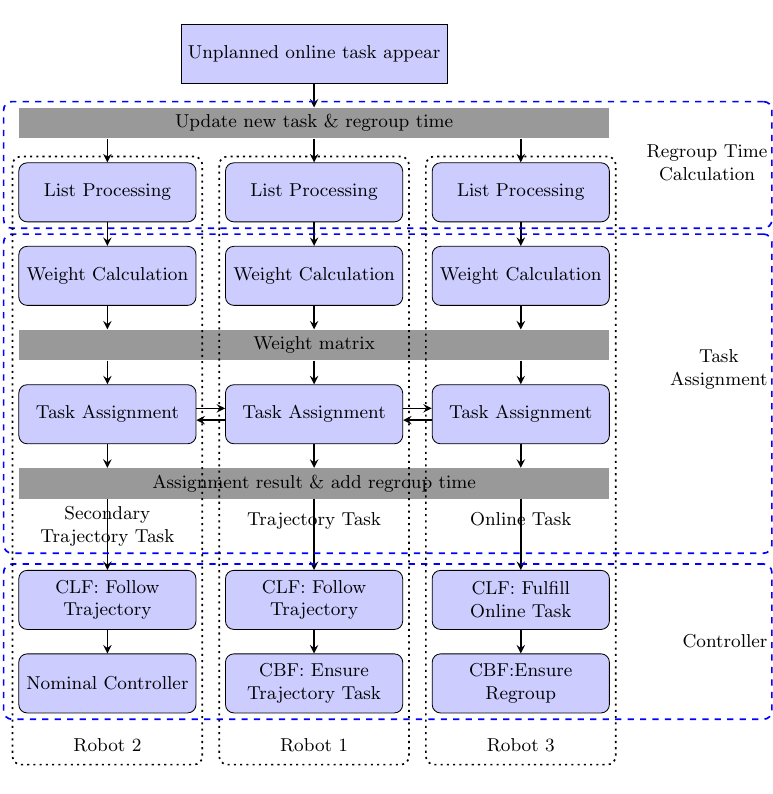}
    \caption{Strategies for co-observation-secured sub-teams encountering online tasks}\label{fig:Flowchart}
\end{figure}

The main contributions of this paper are as follows:
\begin{enumerate}
      \item We present a computationally inexpensive distributed online task assignment algorithm that can handle different priority constraints, such as ensuring the completion of pre-planned tasks while accommodating online tasks. This algorithm leverages an inexact ADMM to solve the task assignment problem in a distributed manner. Compared to existing optimization-based approaches~\cite{chakraa2023optimization}, our method is more scalable for the types of problems discussed in this paper.
     \item We demonstrate the applicability of the task assignment algorithm through a co-observation-secured multi-robot map exploration task. Our task assignment algorithm allows a team of robots with security-critical plans to efficiently handle online tasks without compromising security. The strategy contains three components: the decision-making process for secure deviations and regroups for online tasks, the distributed task assignment algorithm, and CLF and CBF based control framework~\cite{lindemann2018control,article09} to guide robots to meet their spatio-temporal schedule. 
\end{enumerate}

The organization of this paper is as follows. \Cref{sec:preliminary} introduces necessary definitions and preliminaries for task assignment problems. \Cref{sec:task_assignment} presents the distributed task assignment algorithm in detail. \Cref{sec:application} presents a co-observation-secured application. In this application, the capability to deal with online tasks is demonstrated using the proposed algorithm combined with additional decision-making processes and control methods. \Cref{sec:simulation} presents the simulation results. \Cref{sec:Conclusion} concludes this paper.

\section{Definitions and Preliminaries}\label{sec:preliminary}
\subsection{Pre-planned MRS trajectory and sub-teams}
We consider a multi-robot system with $N>1$ robots in an $m$ dimensional workspace, where a partition $\cI=\bigcup_p \cI_p$ is defined to separate the system into $N_p$ \emph{sub-team}s. 
A planner~\cite{yang2021multi} is used to generate reference paths that optimize high-level objectives, while also satisfying security requirements such as enforcing a \emph{co-observation schedule}, where agents are required to regroup at designated locations and scheduled times to maintain situational awareness and detect potential attacks. For example, in \cref{fig:Example_case_co-observation}, the dotted circles indicate locations where robots from the black and green teams are expected to arrive and perform co-observation, as specified in the pre-planned trajectories.

 The planner provides $N_P<N$ trajectories $\{\vq_p\}_{p=1}^{N_p}$ for each sub-team $\cI_p$ that consist of reference waypoints $\vq_p = [q_{p1},\dots,q_{pT}]$ where $q_{ij}\in\real{m}$ is the reference waypoint of team $i$ in an $m$ dimensional state space, $T$ is the time horizon. Robots in the same sub-team $\cI_p$ share the same nominal trajectory $\vq_p$ and treat the trajectory following as their \emph{primary task}. Each sub-team performs the task assignment individually. For the rest of this paper, we only consider states of members in one sub-team instead of the whole MRS system. For sub-team $\cA=\bigcup_i \cA_i$, let $\cA_i$ denotes agent $i$ in the sub-team. Subscript $\circ_{\cA_i}$ is used to indicate the property corresponding to $\cA_i$. For example, the state $x_{\cA_i}\in \real{m}$ represents the current position of \vspace{0.5cm}the agent $\cA_i$ in sub-team.

\subsection{Communication graph}
We assume that agents in the same sub-team $p$ can securely exchange information according to an undirected communication graph $G_p$. To each graph is associated a \emph{Laplacian matrix} $L_p$ defined as follows.
\begin{definition}\label{def:laplacian_matrix} Let $A$ be the adjacency matrix of a symmetric network. Let $D=\diag(A\vct{1})$. Thus, the Laplacian matrix is defined as $L=D-A$.~\cite{chung1997spectral}
\end{definition}
Assuming that the network graph is always connected and symmetric, we have the following:
\begin{lemma}\label{lemma:laplacian properties}
$L\vct{1}=0$ and $L\succ 0$.
\end{lemma}
Let $\alpha=\stack(\{\alpha_i\})$ be a vector of scalar states at each node. The following observation is used further below to identify computations that can be implemented in a distributed manner.
\begin{lemma}
    The computation of $L\alpha$ can be done in a distributed way: each agent $i$ only needs to know $\alpha_i$ and exchange a single round of communication with its neighbors $N_i$.
\end{lemma}
\begin{proof}
    Given the definition of $L$, for an individual agent, the operation can be written as $[L\alpha]_i=\sum_{\in N_i} (\alpha_j-\alpha_i)$, hence the claim.
\end{proof}
\subsection{Key Notations}
 For a matrix $\bullet\in\real{m\times n}$, $\bullet_{:j}\in\real{m}$ now used to represent task-wise (column) variables, while $\bullet_i = \bullet_{i:}\transpose\in\real{n}$ represent agent-wise (row) variables.


\section{Distributed task assignment protocol}\label{sec:task_assignment}
In this section, we formulate the task assignment problem and solve it in a distributed manner using an inexact ADMM-based approach. The problem involves two types of assignable tasks.

\begin{description}
\item[\emph{Trajectory task $P$}]: This task corresponds to the original pre-planned mission trajectory, with designated locations and pre-determined times for the agents to reach. The trajectory task is fulfilled by following the nominal trajectory $q_p$ given to the sub-team. These tasks have the highest priority and are required to have at least one agent in each sub-team always assigned to them. 
\item[\emph{Online tasks $\{O_j\}$}]: These are tasks that are not known at planning time. An online task $j$ appears at a certain time $t_{O_j}$ at a location $q_{O_j}$ (\cref{fig:Example_case_task_appear}), and it is satisfied when a robot reaches a small neighborhood of $q_{O_j}$. We consider that an online task is optional, and may be unsatisfied if there are no available agents, or if the overall security requirements cannot be maintained.
\end{description}

Inside the sub-team, each robot is either assigned to the trajectory task $P$, or to an online task $O_j$. We introduce a matrix of \emph{splitting coefficients} $ \pmb\alpha_{\cA} =[\alpha_{i,j}]$, which are functions of time, and where each element $\alpha_{i,j}=1$ represents the assignment of agent $i$ to task $j$. We formally define the main task assignment problem via the following Linear Program:
\begin{subequations}\label{eq:original_problem}
\begin{align}
    \min_{\alpha} &-\trace(\pmb W_{\cA}\transpose \pmb\alpha_{\cA}),\\
    \subjectto & \sum_{j} \alpha_{ij}=1, \label{eq:agent total}\\
    &\sum_i \alpha_{iP}= 1, \label{eq:primary task}\\
    &\sum_i \alpha_{ij}\leq 1, \label{eq:optional tasks}  \quad \textrm{for}\quad j\neq P\\
    & 0\leq \alpha_{ij} \leq1, 
\end{align}
\end{subequations}
where $\pmb W_{\cA}, \pmb\alpha_{\cA} \in\real{\abs{\cA}\times \abs{\cT}}$ is a matrix of weights, $j\in\cT$ $i\in \cA$, $\cT \subseteq P \union \{O_j\}$ is the set of trajectory and online tasks, and $\abs{\cA}$ and $\abs{\cT}$ are the total number of agents and tasks respectively. 

The rationale behind the constraints is as follows:
\begin{description}
\item[\eqref{eq:agent total}] ensures that each agent is assigned to exactly one task. This is a local constraint that is easy to handle through a local step of ADMM.
\item[\eqref{eq:primary task}] ensures that the trajectory task is always assigned exactly once.
\item[\eqref{eq:optional tasks}] ensures that online tasks do not get assigned to more than one agent. When $\abs{\cA} = \abs{\cT}$, this constraint can be written as $\sum_i a_{ij} = 1$
\end{description}
Each entry $\pmb W_{\cA} = [w_{ij}]$ represents the weight for task $j$ and $\cA_i$; $w_{ij}$ is determined from the current state (location) $x_{\cA_i}$ of agent $i$ as follows:
\begin{description}
    \item[Trajectory task $P$]: 
    \begin{equation}\label{eq:trajectory_task_weight}
    w_{\cA P}(x_{\cA_i})=(d( q_t,x_{\cA_i})+\epsilon)\inverse, 
    \end{equation} where $\epsilon$ is a perturbation term to avoid singularity.  The weight of the trajectory task is formulated as the Euclidean distance $d(q_t,x_{\cA_i})$ between the current location $x_{\cA_i}$ and reference waypoint $ q_t$ at next time step.
    \item[Online tasks $\{O_1,O_2,\ldots\}$]:
    \begin{equation}
        w_{\cA O_j}(x_{\cA_i})=(d(q_{O_j},x_{\cA_i})+\epsilon)\inverse.
    \end{equation} 
    Weight of the online task $O_j$ is formulated as the distance between the current location $x_{\cA_i}$ and task location $q_{O_j}$.
\end{description}
Notice that, in all cases, the weights satisfy the property $w_{ij}\leq \epsilon\inverse$.

Our goal is to compute the coefficients $\vct{\alpha}$ in a distributed manner.
 We first show a distributed solution for problem \eqref{eq:original_problem} when we have an equal number of agents and tasks ($\abs{\cA} = \abs{\cT}$); we extend the solution for other cases ($\abs{\cA} \neq \abs{\cT}$) in later sections.

\subsection{Distributed optimization problem formulation for the square case ($\abs{\cA} = \abs{\cT}$)}\label{sec:square_case}

When the number of agents and tasks (including both trajectory and online) is the same, i.e.,  $\abs{\cA} = \abs{\cT}$, the problem can be written in matrix form (dropping the subscript $\cA$ to simplify the notation) as:
\begin{subequations}\label{eq:nXn_problem}
\begin{align}
    \min_{0\leq \pmb\alpha \leq 1} &-\sum_i (\vw_i\transpose  \pmb\alpha_i)\label{eq:square-problem-objective}\\
    \subjectto & \pmb\alpha \vct{1}=\vct{1} \label{eq:one_task_assigned}\\
    & \vct{1}\transpose \pmb\alpha_{:j}= 1 \quad
\forall j\in \cT \label{eq:one_agent_assigned}
\end{align}
\end{subequations}
where $\pmb\alpha_{:j}$ is the $j$-th columns of $\pmb\alpha$, and $\pmb\alpha_{i} = \pmb\alpha_{i:}\transpose$ denotes the $i$-th row of $\pmb\alpha$,  representing the local coefficient for $\cA_i$.

Although \eqref{eq:nXn_problem} cannot be solved entirely locally, our strategy is to first decompose it using ADMM framework into a separable sub-problem and a non-separable sub-problem, optimized over $\pmb\alpha$ and $\vz\in\real{\abs{\cA}\times \abs{\cT}}$ respectively, and connected two sub-problem via the consensus constraint $\pmb\alpha = \vz$. We then show that the non-separable sub-problem can also be solved in a distributed manner through a modified update rule introduced later in \cref{sec:distributed_computation_of_z}.
First, we define the following: 
\begin{align}
&f_i( \pmb\alpha_i) = - \vw_{i}\transpose  \pmb\alpha_i \label{eq:individual_problem}\\
&\textrm{dom}f_i = \{  \pmb\alpha_i|  \pmb\alpha_i\transpose\vct{1}=1,  0\leq \pmb\alpha \leq 1\} \label{eq:individual_constraint}
\end{align}
\begin{equation}\label{eq:consensus_problem}
	g(\vz)=\begin{cases}
	0 & \textrm{ if } \vct{1}\transpose \vz  = \vct{1}\\
	-\inf & \textrm{ otherwise.}
	\end{cases}
\end{equation} 
where $f_i$ is the locally separated objective function \eqref{eq:square-problem-objective} for $\cA_i$; $\textrm{dom} f_i$ represents the separable constraints \eqref{eq:one_task_assigned} encoded as the domain of $f_i$; and $g(\vz)$ is the indicator function over the feasible set defined by the non-separable constraint \eqref{eq:one_agent_assigned}. 
\begin{remark}
    Term $\textrm{dom} f_i$ and $g(\vz)$ allow the problem to be written compactly in the standard ADMM form. Whereas in the ADMM update steps, $\textrm{dom} f_i$ and $g(\vz)$ are equivalently enforced as conventional constraints during $\pmb\alpha$-updates and $\vz$-updates. 
\end{remark}
Following the ADMM framework \cite{Boyd2011}, the overall problem can be written as:
\begin{subequations}\label{eq:square-problem}
\begin{align}
\min_{\pmb\alpha} \qquad & \sum_i f_i( \pmb\alpha_i) + g(\vz)\\
\subjectto & \pmb\alpha = \vz \label{eq:consensus constraint}
\end{align}
\end{subequations}

The Lagrangian of the problem can be written as: 
\begin{equation}\label{eq:Lagrangian}
	\cL(\pmb\alpha,\vz,\vu) = \sum_i f_i( \pmb\alpha_i) + g(\vz) + (\rho/2) \norm{\pmb\alpha-\vz+\vu}_2^2
\end{equation}
 where $\rho >0$ is the penalty parameter which controls the dual update step and represents the penalty on primal feasibility violation. We keep $\rho$ constant in this work for simplicity, as tuning $\rho$ is not the focus of this paper. $\vu\in\real{\abs{\cA}\times \abs{\cT}}$ is the scaled dual variable in the ADMM framework, interpreted as the accumulated residual of the primal feasibility constraint $\pmb\alpha=\vz$.



The ADMM framework then proceeds with iterations through the following steps:
\begin{enumerate}
	\item $\pmb\alpha$-update:
	\begin{equation}
	\begin{aligned}
	\pmb\alpha \gets \argmin_{\pmb\alpha} &-\sum_{i=1}^n \vw_{i}  \pmb\alpha_i\transpose + (\rho/2) \norm{\pmb\alpha-\vz^k+\vu^k}_2^2\\
	\subjectto &\pmb\alpha\transpose\vec{1}=1\\
	&\vec{0}\leq\pmb\alpha\leq\vec{1}
	\end{aligned}
	\end{equation}
	The optimization problem is separable, and the local objective for agent $i$ can be written as a QP problem:
	\begin{equation}
	\begin{aligned}
	\min_{ \pmb\alpha_i} & \quad \frac{1}{2} \pmb\alpha_i\transpose  \pmb\alpha_i + (-\vz_i^k+\vu_i^k-\frac{\vw_{i}}{\rho})\transpose \pmb\alpha_i\\
	\subjectto &\quad \vct{1}\transpose \pmb\alpha_i=1\\
	&\quad \vec{0}\leq \pmb\alpha_i\leq\vec{1}.
	\end{aligned}
	\end{equation}
	Each agent can compute its part of the solution independently.
	
    \item $\vz$-update:
	\begin{equation}\label{eq:z_update}
	 \vz \gets \argmin_{\vct{1}\transpose \vz  = \vct{1}} \sum_{j= 1}^m\frac{1}{2}\norm{\vz-(\pmb{\alpha}^{k+1}+\vu^k)}^2
	\end{equation}
	
	This problem is separable over tasks $\vz_{:j}$, i.e. the columns of $z$. Each sub-problem (i.e., the solution for task $j$) can be formulated as:
	\begin{subequations}\label{eq:z_update_reform}
	\begin{align}
	\min & \frac{1}{2}\norm{ \vz_{:j}-\pmb\alpha_{:j}^{k+1}- \vu_{:j}^k}^2\label{eq:z_update_reform_obj}\\
	\subjectto & \vct{1}\transpose \vz_{:j}=1\label{eq:z_update_reform_constraint}
	\end{align}
	\end{subequations}
	where $\vz_{:j}$, $\pmb\alpha_{:j}$ and $\vu_{:j}$ are all $j$th column of the corresponding matrices.
	The objective can be written as $\sum_i (\vz_{ij}-\pmb\alpha_{ij}-\vu_{ij})^2$, which is separable over agents, but the constraint is not, so we cannot directly solve this step in a distributed way. However, we show in Section \ref{sec:distributed_computation_of_z} that this can be solved using \emph{projected gradient descent}.
\item $\vu$-update:
\begin{equation}
\vu_i \gets \vu_i^k +  \pmb\alpha_i^{k+1} - \vz_i^{k+1}
\end{equation}
Each multiplier $\vu$ update only needs information that is local to the node, so this step can be executed independently at each node. 
\end{enumerate}

\subsection{Distributed $\vz$-update}\label{sec:distributed_computation_of_z}
This subsection introduces the implement how the proximal update helps to solve~\eqref{eq:z_update} in a distributed fashion. When the communication graph $G_p$ is connected,~\eqref{eq:z_update_reform} can be approximately solved using a distributed projected gradient descent, where each robot requires only information from its neighbors.


\begin{lemma}\label{lemma:keep_constraint_valid}
The iterations $y[k+1]=y[k]+Lv[k]$ satisfy $\vct{1}y[k]=\vct{1}\transpose y[0]$ for all $k$. 
\end{lemma}
\begin{proof} Using the properties of the Laplacian matrix in Lemma \ref{lemma:laplacian properties}, we have the following invariant:
$\vct{1}\transpose y[k+1]=\vct{1}\transpose y[k]+\vct{1}\transpose Lv[k]=\vct{1}\transpose y[k]$, from which the claim follows.
\end{proof}

The \emph{gradient} of problem (\ref{eq:z_update_reform}) can be computed in a distributed manner, we simply have
$\grad_{\vz_{:j}}\cL=(\vz_{:j}^k-\pmb\alpha_{:j}^{k+1}-\vu_{:j}^k)$.

Thus, by iterating 
\begin{multline}\label{eq:z_update_distributed}
\vz_{:j}[k'+1]=\vz_{:j}[k']-\varepsilon L\grad_{\vz_{:j}} \cL\\
=\vz_{:j}[k']-\varepsilon L (\vz_{:j}[k']-\pmb\alpha_{:j}^{k+1}-\vu_{:j}^k)\\
=(I-\varepsilon L)\vz_{:j}[k']+\varepsilon L(\pmb\alpha_{:j}^{k+1}+\vu_{:j}^k)
\end{multline}
we show that
\begin{itemize}
	\item $\vz[k']$ converges to the desired solution of \eqref{eq:z_update};
	\item Each step is distributed.
\end{itemize}

Notice that the update for each node can be written as the sum of contributions over its neighbors. Specifically, for any state vector $\vv$, from~\cref{def:laplacian_matrix}, we have $L\vv=(D-A)\vv=D\vv-A\vv$. Updates in the form of $L\vv$ can be written individually for each agent $\cA_i$ as $i$-th entry of $L\vv$, which yields:
\begin{equation}
[L\vv]_i = d_iv_i-\sum_{i'} a_{ii'} v_{i'}=\sum_{(i,i') \in E_i} (v_i - v_{i'}),
\end{equation}
where $d_i=\sum_{i':(i,i') \in E} a_{ii'}$, and $E_i$ denotes the set of edges connected to $\cA_i$.

 In contrast to the task-wise update in \eqref{eq:z_update_distributed}, we provide the agent-wise local update for each agent $\cA_i$. Let $\vz_i = \vz_{i:}\transpose$ denote the $i$th row of $\vz$, the update law is given by:
\begin{multline}
\label{eq:z_update_distributed_agent}
 \vz_{i}[k'+1]= \vz_{i}[k']-\varepsilon\sum_{(i,i') \in E}\Bigl((\vz_{i}[k'] - \vz_{i'}[k])\\
 -(\pmb\alpha^{k+1}_{i}-\pmb\alpha^{k+1}_{i'})-(\vu^k_{i} - \vu^k_{i'})\Bigr) 
\end{multline}
Note the different $\pmb\alpha^{k+1}$ and $\vu^k$ are constant while we update the $\vz_i$. 

In \cref{app:proof}, we show that a single iteration of~\eqref{eq:z_update_distributed} is enough to guarantee $z$ converge to its optimal value. While in practice, \eqref{eq:z_update_distributed} can be repeated multiple times to accelerate convergence. The number of iterations for~\eqref{eq:z_update_distributed} is a trade-off between computation time and communication overhead. For instance, simulations in~\cref{sec:simulation} show that five iterations of the $z$-update represent an effective balance in the particular setting considered.

\begin{remark}\label{rmk:not-binary}
    The assignment problem requires multiple iterations to converge; during this process, $\alpha$ isn't necessarily $0$s or $1$s. Rounding $\alpha$ prematurely may cause chattering—frequent switching between assignment states. In practice, we interpret $\alpha$ as priority coefficients in the control layer, as detailed in \cref{sec:control-design}, \cref{rmk:alpha-in-control}.
\end{remark}


\subsection{Shadow agents and secondary trajectory tasks}
The distributed solution of problem~\eqref{eq:nXn_problem} is based on \cref{lemma:keep_constraint_valid} which requires that $\pmb\alpha$ is a square matrix and $\vct{1}\transpose \pmb \alpha = \vct{1}$. 
This cannot be the case when the number of robots and tasks are not equal (i.e. $\abs{\cA} \neq \abs{\cT}$). 
For these cases, we introduce the concepts of \emph{secondary trajectory task} (when $\abs{\cA}>\abs{\cT}$) and \emph{shadow agent} (when $\abs{\cA}<\abs{\cT}$), which allow us to reformulate the general assignment problem \eqref{eq:original_problem} to square case introduced in \cref{sec:square_case}. 

\subsubsection{Secondary trajectory tasks}

For cases where the number of agents is larger than the number of tasks (i.e. $\abs{\cA} > \abs{\cT}$), we introduce \emph{secondary trajectory tasks} $P'=\bigcup_j P'_j$, where $\abs{P'} = \abs{\cA} - \abs{\cT}$, to capture the marginal gains of having more than one agent following the planned trajectory. The modified weight matrix is denoted as $\tilde{\pmb W}$.

To ensure correctness, the weight for $P'$ should be lower than all the weights of the online tasks, as well as the trajectory task $P$; specifically, we propose $w_{iP'_j}(x_{\cA_i})=(d(P,x_{\cA_i})+\epsilon')\inverse$, where $\epsilon'\gg\epsilon$, and $\epsilon'$ is significantly larger than the maximum distance between any agent and objective targets (i.e. $ \epsilon' \gg \max_{\forall O,\forall x}d(O,x)$).
\begin{proposition}
    The introduction of $P'$ will not alter the assignment of $P$ and $O$ in the original problem.
\end{proposition}
\begin{proof}
    Let $\pmb W, \pmb \alpha_o$ be the weights and optimal solution for the rectangular problem~\eqref{eq:original_problem}, and $\tilde{\pmb W},\tilde{\pmb\alpha}$ be the set of weights and optimal solution for the square problem~\eqref{eq:nXn_problem}. Note that $W$ is a subset of the entries of $\tilde{\pmb W}$. Let $\tilde{\pmb \alpha}_o$ be the result $\pmb \alpha_o$ extended in a way such that the unassigned agents are assigned to $P'$; note that $\pmb \alpha_o$ is a subset of $\tilde{\pmb \alpha}_o$.

    To prove the claim, we want to prove that $\tilde{\pmb \alpha}=\tilde{\pmb \alpha}_o$.
    Then, by way of contradiction, assume that there exist a solution of $\tilde{\pmb \alpha}\neq \tilde{\pmb \alpha}_o$ that gives a better cost, i.e., 
    \begin{equation}
        \trace(\tilde{\pmb W}\transpose\tilde{\pmb \alpha})>\trace(\tilde{\pmb W}\transpose\tilde{\pmb \alpha}_o)\label{eq:assumption for contradiction}
    \end{equation}
    and the two solutions differ from each other only by two agents: $\cA_1$ is unassigned in $\alpha_o$, but is assigned to $P$ in $\tilde{\alpha}$; conversely, $\cA_2$ is assigned to $P$ in $\alpha_o$, but is assigned to $P'_1$ in $\tilde{\alpha}$. This yields:
    \begin{multline}\label{eq:counterproof2}
        \trace(\tilde{\pmb W}\transpose \tilde{\pmb\alpha})\geq tr(\tilde{\pmb W}\transpose \tilde{\pmb\alpha}_o) \Longrightarrow \\  (w_{P}(x_{\cA_1})+w_{P'_1}(x_{\cA_2}))\geq(w_{P'_1}(x_{\cA_1})+w_{P}(x_{\cA_2})).
    \end{multline}

   However, from the fact that $\alpha$ is optimal for problem~\eqref{eq:original_problem}, $w_{P}(x_{\cA_1}) < w_{P}(x_{\cA_2})$; and from the assumption of $\epsilon\gg\epsilon$, $w_{P'_1}(x_{\cA_1})$ and $w_{P'_1}(x_{\cA_2})$ are sufficiently small and do not alter this inequality. This leads to a contradiction for~\eqref{eq:counterproof2}, hence $\alpha'=\tilde{\alpha}$ and the claim is proven.  Similarly, since the weights of secondary trajectory tasks are sufficiently small, any case involving an alteration in the assignment of $O_n$ or changes to more than one pair of assignments will contradict the optimality of problem~\eqref{eq:original_problem}.
\end{proof}
\subsubsection{Shadow agent}
For cases where the number of agents is smaller than the number of tasks (i.e. $\abs{\cA} < \abs{\cT}$), we introduce a set of \emph{shadow} agents $\cS^{\cA_i}=\bigcup_l \cS^{\cA_i}_{l}$, where $\abs{\cS^{\cA_i}} = \abs{\cT}-\abs{\cA}$ and $\cS^{\cA_i}_{l}$ are virtual replicas of $\cA_i$ that are considered to have the same location and graph connectivity as $\cA_i$ (communications and computations are handled by the physical $\cA_i$). All the agents in $\cS^{\cA_i}$ can always communicate with physical agent $\cA_i$, and vice versa. If $\cA_i$ is connected with another physical agent $\cA_j$, then all the agents in $\cS^{\cA_i}_{i}$ can communicate with $\cA_j$. There might be more than one shadow agent for a single physical agent.

Assigning tasks to shadow agents means that the corresponding tasks are abandoned, and will not be performed by an actual agent. For simplicity, we let $\cS_l$ denote $\cS_l^{\cA_i}$. We define the weights of the shadow agents as 
\begin{description}
	\item[Trajectory task $P$]: $w_{SP}(x_{\cS_{l}})=-M$, where $M>2\epsilon\inverse$. Effectively, this implies that it is never advantageous, from the optimization perspective, to assign a shadow agent to the trajectory task.
	\item[Online tasks$\{O_j\}$]  $w_{SO_j}(x_{\cS_{l}})=-(d(O_j,x_{\cS_{l}})+\epsilon)\inverse$ This implies that assigning tasks to shadow agents (i.e., abandoning the tasks), from the optimization perspective, is never encouraged; moreover, farther tasks are more likely to be abandoned.
\end{description}

We now show that solving the problem with \emph{shadow agent}s is equivalent to solving it without. Since constraint~\eqref{eq:one_agent_assigned} only ensures that there exists an agent $\cA$ (among regular agents $\cA$ or shadow agents $\cA_S$) that is assigned to $P$; we want to show that $a\in \cA$. 

\begin{proposition}
    Agents in $\cA_S$ will never be assigned to $P$ if $M>2\epsilon\inverse$.
\end{proposition}
\begin{proof}    
   By way of contradiction, assume $\pmb\alpha$ be an optimal solution where $\cS_{l}$ is assigned to $P$ with weight $w_{SP}(x_{\cS_{l}})=-M$. Correspondingly, the real agent $\cA_{i}$ corresponding to $\cS_{l}$ is assigned with an online task $O_j$, with weight $0<w_{\cA O_j}(x_{\cA_i})\leq \epsilon\inverse$. Consider a different solution $\tilde{\pmb{\alpha}}$, where $\cS_{l}$ is assigned to $O_j$ with weight $w_{SO_j}(x_{\cS_l})=w_{\cA O_j}(x_{\cA_i})$ and $\cA_{i}$ is assigned to $P$ with weight $w_{AP}(x_{\cA_i})$. Then we have
    \begin{multline}
        \trace(W\transpose \tilde{\pmb\alpha})-\trace(W\transpose\pmb\alpha) \\ 
        =w_{AP}(x_{\cA_i})-w_{SO_j}(x_{\cS_l})-(-M+w_{\cA O_j}(x_{\cA_i}))\\    
        =(M+w_{AP}(x_{\cA_i}))-2w_{\cA O_j}(x_{\cA_i}) \geq M-2\epsilon\inverse.
    \end{multline}
    Since $M>2\epsilon\inverse$, $\trace(W\transpose \pmb\alpha')-\trace(W\transpose\pmb\alpha)>0$, hence $\pmb\alpha$ is not optimal, leading to the contradiction. 
\end{proof}

\section{Online tasks in Co-observation secured map exploration task}\label{sec:application}
In this section, we apply the proposed task assignment algorithm within the co-observation-secured application introduced earlier in~\cref{sec:introduction}. This application builds on previous work in the field \cite{yangmulti,wardega2019resilience} to mitigate risks posed by physical masquerade attacks in multi-robot systems (MRS). Our solution incorporates co-observation and reachability constraints to ensure security while allowing robots to handle both pre-planned and online tasks. As illustrated in~\cref{fig:Flowchart}, the strategy is structured into three main sections:
\begin{figure*}[h]
	\centering
    \subfloat[Online tasks appear\label{fig:Example_case_task_appear}]{\includegraphics[width=0.28\linewidth]{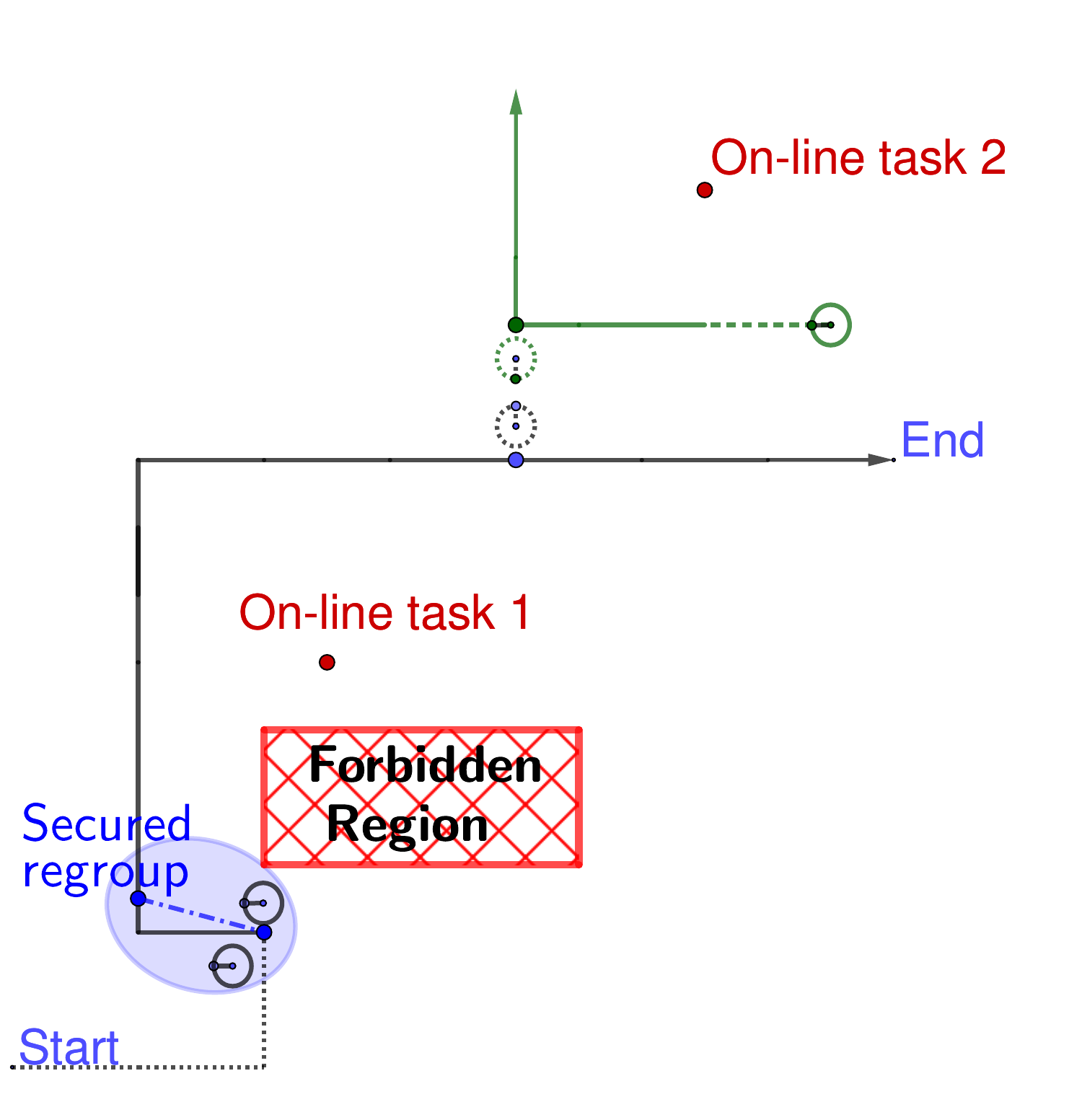}}
    \subfloat[Online task 1 assigned\label{fig:Example_case_task_assigned}]{\includegraphics[width=0.28\linewidth]{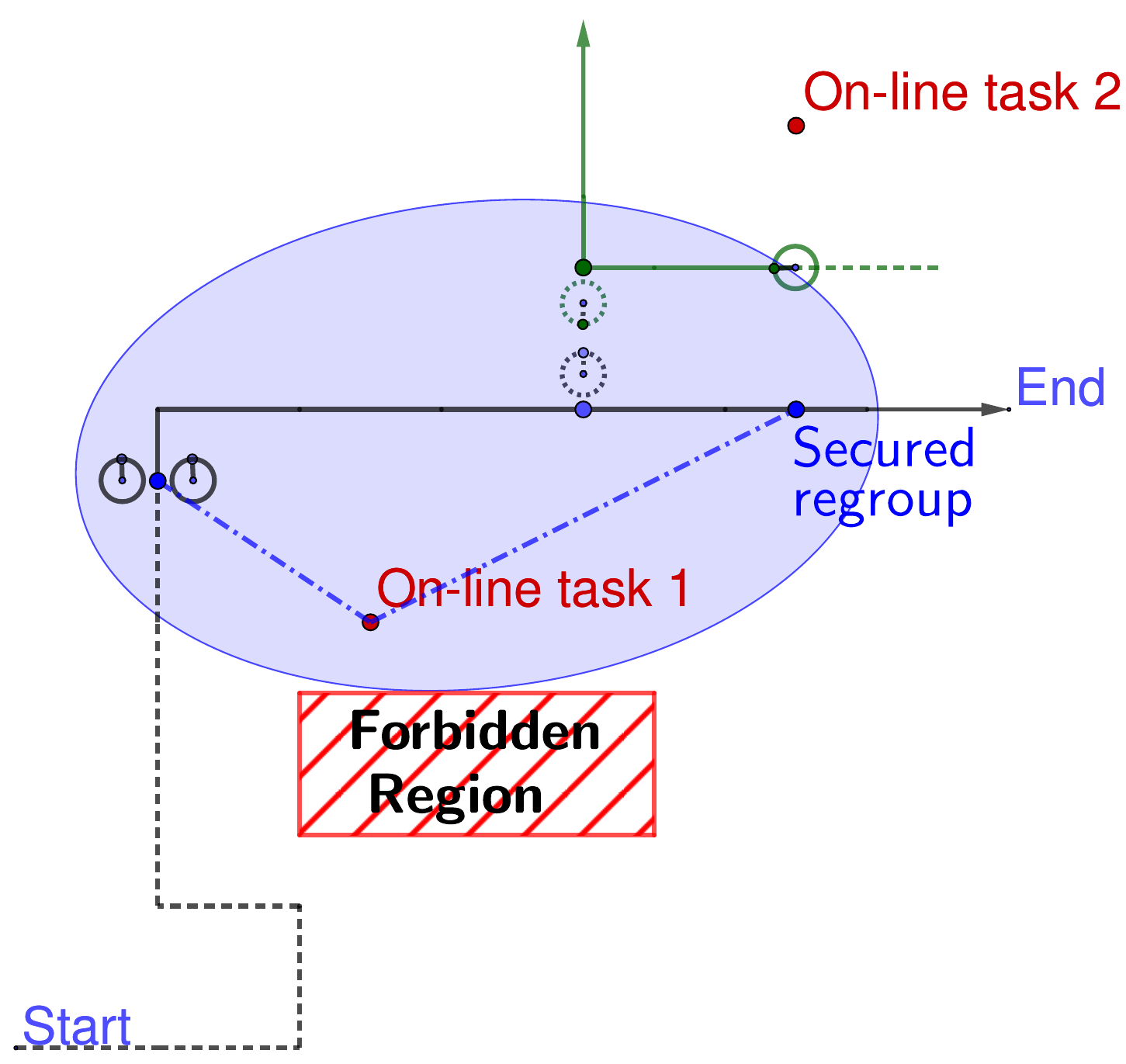}}
    \subfloat[Co-observation\label{fig:Example_case_co-observation}]{\includegraphics[width=0.28\linewidth]{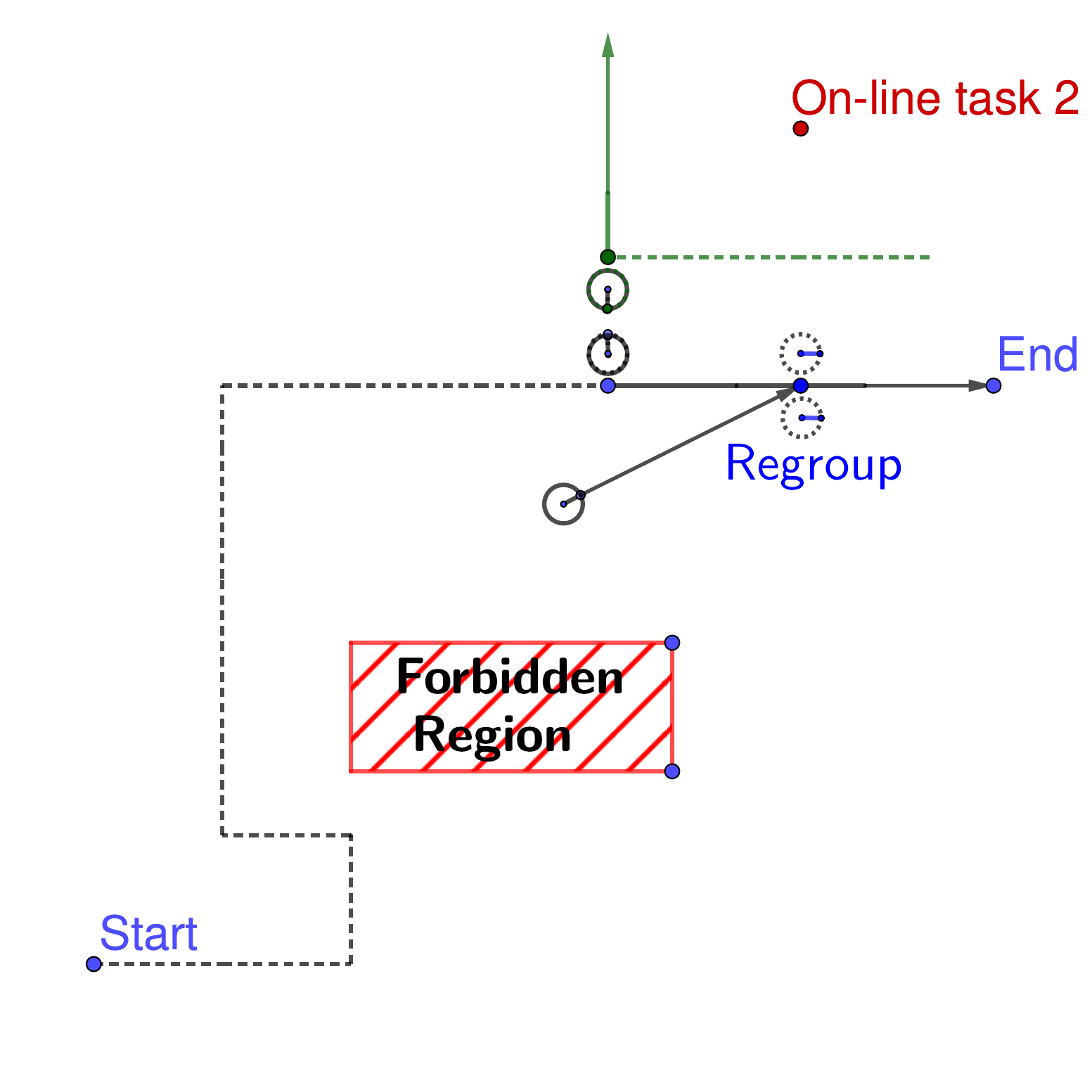}}
    \caption{ Online task 1 and 2 appear during the mission. Solid circles represent the current positions of agents, while dotted circles indicate locations where agents are expected to appear for co-observation. (\ref{fig:Example_case_task_appear}) Green team only have no extra robot available to assign online task 2 with. The latest regroup time for black team do not allow a secure deviation for available online tasks. Thus, no task is assigned. (\ref{fig:Example_case_task_assigned},\ref{fig:Example_case_co-observation}) A safe regroup time is found for black team to fulfill online task 1, one robot got assigned with trajectory task for co-observation with green team. The other robot deviated for online task 1 is required to regroup at the pre-defined regroup time and location, and co-observe with the robot with trajectory task. }\label{fig:5_Example_set}
    \label{fig:Simulation-results}
\end{figure*}

\begin{itemize}
    \item \textbf{Regroup Time Calculation:} Determine whether an unassigned online task can be fulfilled while still satisfying a security condition given by \emph{reachability ellipsoid}; if this is possible, it computes the latest regroup time; this is illustrated in \cref{sec:regroup_computation}.
    \item \textbf{Task Assignment Algorithm} Apply the algorithm introduced in~\cref{sec:task_assignment} to find assignments of agents to trajectory and online tasks via distributed computations. 
    \cref{sec:task-assignment-security} discusses the reasons why solutions to the assignment never break the security guarantees given by co-observations.
    \item \textbf{Online Control Scheme} Use the \emph{CLF-QP} algorithm to compute reference control inputs and use CBF for \emph{Signal Temporal Logic (STL)} tasks as a security filter to guarantee the security of the system. 

    This part transforms the time and locations produced by the online assignment into \emph{spatial constraints} (follow a trajectory, or reach a given location) and \emph{timing constraints} (to avoid missing co-observation times).
    These constraints and their application to the real-time control of agents, are discussed in \cref{sec:control-design}
\end{itemize}
In this section, we first introduce the preliminaries in~\cref{sec:definition}; then the regroup time calculation is introduced in~\cref{sec:regroup_computation}; the CBF-CLF based online control scheme is introduced in~\cref{sec:CBF}.
\subsection{Preliminaries for co-observation-secured planning}\label{sec:definition}
In this section, we define the potential attacks faced by the MRS in the applied scenario and introduce the corresponding security requirements considered in the path planning phase. These requirements are then integrated into the online control problem in later sections to ensure the secure execution of unplanned online tasks.

\subsubsection{Forbidden region}
As part of our scenario, we assume that the environment contains \emph{forbidden regions} $\cF = \bigcup_k \cF_k$ in which none of the robots should enter (e.g. because they contain sensitive information or human operators). Forbidden region are modeled as convex polygons (shown as the red rectangle in~\cref{fig:5_Example_set}; one can use multiple, possibly overlapping, polygons for non-convex regions).

\subsubsection{Plan-deviation attacks}\label{app:attack-model}
We consider \emph{plan-deviation attacks} as introduced by~\cite{wardega2019resilience} to model potential threats to multi-agent systems. Specifically, we assume that a robot in the system has been compromised by an attacker who intends to violate the security constraints by entering forbidden areas undetected.
The attacker has full knowledge of the motion plan and aims to masquerade the compromised robot as a legitimate one.
The attacker intends to let the compromised robot perform deviations from the nominal plan and seek access to forbidden areas. 
We refer to these malicious deviations as plan-deviation attacks. An undetected plan-deviation occurs when a compromised robot deviates from the motion plan while providing a false self-report to the system about its location. Under our model, we consider such deviations to go undetected by the system as long as the self-reports of all other robots remain unchanged.

\subsubsection{Co-observation schedule}\label{sec:coobservation_intro}
The self-reports within the MRS is no longer reliable considering the potential of plan-deviation attacks. In previous works~\cite{wardega2019resilience,yang2021multi}, we proposed to leverage onboard sensing capabilities to integrate a mutual observation plan into the multi-agent trajectory. 
We assume that we are given a \emph{co-observation schedule}, a sequence of waypoints where two or more sub-teams of robots are required to meet (shown as dotted robots in~\cref{fig:5_Example_set}) such that, for each sub-team, any faulty or attacking agent breaching security specifications (i.e. in this case, trespassing forbidden region) would inevitably violate the plan, ensuring their actions are detectable. 

\subsubsection{Reachability region}
 Paired with the co-observation schedule, we introduced the concept of the \emph{reachability region} to analyze whether a compromised robot could reach any forbidden area between scheduled co-observation locations. The reachability region is defined as the set of all points in the free configuration space that a robot can feasibly reach while traveling from one co-observation location to the next one within the given time interval. For simplicity, this analysis assumes a robot with a first-order integrator model and a maximum velocity cap $v_{max}$, allowing us to over-approximate the reachability region by a \emph{reachability ellipsoid}.

\begin{definition}\label{sec:ellipsoidal definition}
    Consider a robot $i$ starting from $q_{1}$ at time $t_1$ and reaching $q_{2}$ at time $t_2$. The \textbf{reachability ellipsoid $\cE$} is defined as the region $\cE^{q_2}_{q_1} = \mathcal{E}(q_1,q_2,t_{1},t_{2})=\{\tilde{q}\in\mathbb{R}^n: d(q_1,\tilde{q})+d(\tilde{q},q_2)<2a\}$, where $a=\frac{v_{max}}{2}(t_2-t_1)$, and $d(\cdot,\cdot)$ denotes the Euclidean distance between two points.
\end{definition}

 If the reachability region does not intersect with any forbidden regions $\cE^{q_2}_{q_1}  \cap \cF = \emptyset$, it can be guaranteed that the robot starting at $(q_1,t_1)$ is unable to reach $(q_2,t_2)$ if it trespassed any forbidden region. 
    
Accordingly, we can provide a formal security guarantee against plan-deviation attacks for a multi-agent trajectory.
\begin{definition}\label{rmk:revised-security}
A multi-robot trajectory is secured against plan-deviation attacks if there exists a co-observation plan such that the reachability region between each consecutive co-observation does not intersect with any forbidden regions.
\end{definition}

\subsection{Regroup time Computation}\label{sec:regroup_computation}
When addressing unplanned tasks that require deviations from the planned trajectory, robots must ensure that these deviations do not compromise security.
We leverage the concepts of co-observation schedules and reachability regions to pose constraints on robots deviating and rejoining teams, as made rigorous in the following definition.
\begin{definition}
    A deviation for robot $\cA_i$ in team $\cA$, deviating from the trajectory at $(x_1, t_1)$, and rejoining it at $(x_2, t_2)$, is secured if the following conditions are met: \begin{enumerate}
    \item At least one other robot in the team $\cA$ remains on the planned trajectory between $t_1$ and $t_2$.
    \item The reachability region between $(x_1, t_1)$ and $(x_2, t_2)$  does not intersect any forbidden regions, i.e.,  $\cE^{x_2}_{x_2} \cap \mathcal{F} = \emptyset$.
    \end{enumerate}
\end{definition}

It becomes apparent from the definition that, at every step of a trajectory, an agent is free to deviate if it can rejoin its team quickly enough. We therefore introduce the following.
\begin{definition} 
    The \textbf{latest secured regroup time} for a deviation at $(x_1, t_1)$ is defined as the largest time $t_2\in [t_1,T]$ such that rejoining at $(x_2,t_2)$ can secure this deviation.
\end{definition}

To avoid complex online computation, we propose to map each waypoint $q_t$ with its latest secured regroup time $t^r_{q_t}$.
\begin{definition} 
    The \textbf{latest regroup time lookup table} $\{(q_t ,t^r_{q_t})\}^T_{t=1}$ maps each waypoint $q_t$ on the planned trajectory $\{\vq_p\}$ to its latest secured regroup time $t^r_{q_t}$.
\end{definition}

The lookup table can be built by applying \cref{alg:lookup_table} for each one of the $N_p$ reference trajectories.
\begin{algorithm}
    \caption{Lookup Table construction for sub-team $N_p$}
    \label{alg:lookup_table}
    \begin{algorithmic}[1]
    \Require Planned trajectory $\{\mathbf{q}_p\}$, forbidden regions $\mathcal{F}$
    \State \textbf{Initialization:} Lookup table $\mathcal{L}$
    
    \For{$i = 1$ to $T$} 
        \State $q_1 \gets q_i \in \{\mathbf{q}_p\}$
        \State $j \gets i+1$, $q_2 \gets q_j \in \{\mathbf{q}_p\}$
        \While{Section $(q_i,q_j)$ is secured ($\cE^{q_2}_{q_1} \cap \cF = \emptyset$) }
            \State $j \gets j+1$, $q_2  \gets q_j \in \{\mathbf{q}_p\}$
        \EndWhile
        \State $t_r \gets j-1$
        \State Store $\{(q_i, t_r)\}$ in $\mathcal{L}$
    \EndFor
    \end{algorithmic}
\end{algorithm}

The lookup table is used when an online task $O^n_j$ appears or remains unassigned in a previous time step. Agents check for the latest safe regroup time $t_r$ for the next timestep $t$: if $d(p_t,x_{O^n_j})+d(p_{t_r},x_{O^n_j})>v_{max}(t_r-t)$, $O^n_j$ is outside the reachability region $\cE_t^{t_r}$; thus, deviation for $O^n_j$ can not be secured, $O^n_j$ will not be assigned (\cref{fig:Example_case_task_appear}). Otherwise, $O^n_j$ is passed to the task assignment algorithm (\cref{fig:Example_case_task_assigned}).

Since the lookup table is precomputed and stored locally on each robot, once a task is assigned, all agents (both deviating and non-deviating) will be aware of the regroup time without explicit communications. 

\subsection{Online task assignment problem}\label{sec:task-assignment-security}
As shown in~\cref{fig:Flowchart}, after determined the secured regroup time, available tasks are updated and passed to the task assignment algorithm introduced in~\cref{sec:task_assignment}. The priorities in the assignment ensure that at least one agent will satisfy the trajectory task $P$ to fulfill the co-observation schedule, while the extra robots are assigned to available online tasks $O_j$ (\cref{fig:Example_case_co-observation}). 

Note that the overall security is maintained, and none of the robots will be able to reach the forbidden zones because:
    \begin{itemize}
    \item For the robots assigned to trajectory tasks $P$, this type of deviation would imply breaking the constraints imposed by the co-observation schedule.
    \item For the robots assigned to the online tasks $O_j$, this type of deviation would cause them to miss the implicit co-observation constraints with the rest of the team at the regroup location and time.
\end{itemize}

\subsection{Online Control Scheme}\label{sec:CBF}
The reference trajectories and co-observation requirements are defined in terms of waypoints in discrete time. To enable execution in continuous time for real-world scenarios, we propose an online control framework for each individual robot. In this framework, robots assigned to trajectory and secondary trajectory tasks follow their predefined secure paths, while robots assigned to online tasks move toward the task location, complete the task, and then return to the reference trajectory. Additionally, to maintain security, the framework ensures that robots following trajectory tasks reach their designated co-observation locations within the required time, while robots performing online tasks must return to their trajectory and rejoin the group before the regroup time.

\subsubsection{Robot dynamics}
Consider robots with an affine-input dynamical system
\begin{equation}
\dot{x}=f(x)+g(x)u
\end{equation}
with $f$ and $g \neq 0$ locally Lipschitz continuous, $x \in \mathbb{R}^m$ represents the location of the robot and $u \in U \subset \mathbb{R}^{n}$ is the control input.

\subsubsection{Navigation and timing control constraints via CBF and CLF functions}

All assignments and security requirements can be represented as navigation and timing control constraints, which are formulated via CBF and CLF functions. The definitions of these functions are provided below.

Given a set $\mathcal{C}$ defined as $\cC=\{x\in\real{m},t\in\real{}_{\geq 0}:h(x,t)\geq 0\}$ for a continuously differentiable function $h(x):\mathbb{R}^{m}\times\real{}_{\geq 0} \rightarrow \mathbb{R}$. The function $h$ is called a \emph{CBF}, if there exists a class $\mathcal{K}$ function $\beta $ such that 
\begin{equation}\label{eq:CBF_constraint}
\sup_{u \in U} \frac{\partial h(x,t)}{\partial x}(f(x)+g(x)u)+\frac{\partial h(x,t)}{\partial t} +\beta(h(x,t)) \geq 0.
\end{equation}
In the following content, we denote $\frac{\partial h(x,t)}{\partial x}(f(x)+g(x)u)+\frac{\partial h(x,t)}{\partial t}$ as $\dot{h}(x,t)$ for simplicity. Controllers satisfying the CBF constraint~\eqref{eq:CBF_constraint} ensure that the set $\cC$ remains \emph{forward invariant}. This means that if the initial state satisfies $x(0)\in \cC$, the state remains in $x(t)\in \cC, \forall t>0$.

In this application, we use CBF timing constraints as security filters to ensure that: \begin{itemize}
    \item Robots with trajectory tasks reach their co-observation locations within the required time; 
    \item Robots with online tasks return to their trajectory and rejoin the group before the regroup time. 
\end{itemize}

Similarly, consider a continuously differentiable positive definite function $V : \real{m}\rightarrow \real{}$. $V$ is called a \emph{CLF} if there exist a class $\cK$ function $\gamma$, such that:
\begin{equation}
    \inf_{u\in U} [L_f V(x) + L_g V(x)u + \gamma(V(x))] \leq 0.
\end{equation}
In this application, we use CLF as navigation constraints to ensure that: \begin{itemize}
    \item Robots assigned trajectory and secondary trajectory tasks follow their predefined secure trajectories;
    \item Robots assigned to online tasks move toward the task location, complete the task, and then return to the reference trajectory.
\end{itemize}

\subsubsection{CLF functions for navigation}
We define the CLF as $V_t(x) = d(x,q_i)$ for trajectory task and secondary trajectory tasks, where $q_i \in \vq$ is the planned waypoint for the next timestep. Since the reference trajectory is in discrete time, $V_t$ remain the same for $t\in [i-1,i)$, and switch for the next reference point $V_t(x) = d(x,q_{i+1})$ for $t\in[i,i+1)$. And $V_{O_j}(x) = d(x,x_{O_j})$ is the candidate CLF for online tasks. 

\subsubsection{CBF functions via STL}
Both trajectory and online task security requirements are enforced through co-observations at specific times and locations. For trajectory tasks, co-observations occur between sub-teams based on the co-observation schedule. For online tasks, they occur within a sub-team, ensuring the returning robot is observed at the regroup time along the reference trajectory. This also adds additional co-observation requirements for the trajectory task robot. For online control, we use STL to formally encode these tasks with strict deadlines, ensuring their satisfaction through Control Barrier Functions (CBFs)~\cite{lindemann2018control}.

A predicate $\mu$ is derived from evaluating a function $h:\real{m}\rightarrow \real{}$ as:
\begin{equation*}
    h(x)\coloneqq \begin{cases}
        \text{True}& \text{if } h(x)\geq 0, \\
        \text{False} & \text{if } h(x) <0.
    \end{cases}
\end{equation*}
The STL syntax defines a formula as:
\begin{equation*}
    \phi ::= \text{True} | \mu | \neg \phi | \phi_1 \wedge \phi_2 | \phi_1 \textbf{U}_{[t_1,t_2]} \phi_2,
\end{equation*}
where $\phi_1$ and $\phi_2$ are STL formulas and $t_1,t_2\in\real{}_\geq0$ with $t_2>t_1$. Notice that $\textbf{F}_{[t_1,t_2]}\phi$ and $\textbf{G}_{[t_1,t_2]}\phi$ can be defined in terms of $\textbf{U}$, thus are omitted in this formula.

In this paper, we consider the following STL fragment:
\begin{align}
    &\psi ::= True | \mu | \neg \mu | \psi_1 \wedge \psi_2 ,\\
    &\phi ::= G_{[a,b]} \psi | F_{[a,b]} \psi | \psi_1 U_{[a,b]} \psi_2 | \phi_1 \wedge \phi_2,
\end{align}
where $\psi$ defines state-based conditions that involve only Boolean logic with $\psi_1$ and $\psi_2$ are formulas of $\psi$. While $\phi$ defines temporal properties that are built from class $\psi$, and $\phi_1$ and $\phi_2$ are formulas of the temporal operators.

To better handle these requirements via CBF, we model each individual requirement, i.e. co-observation at location $q$ at time $c$, using STL formulas. The formula $\phi_i$ is defined as:
\begin{equation}
\phi_i := \phi_{i1} \wedge \phi_{i2},
\end{equation}
 where:
 \begin{itemize}
    \item $\phi_1 := F_{[0,c]} \psi$ to ensure that the robot maintains the capability of  reaching the scheduled co-observation location on time,
    \item and $\phi_2 := G_{[c,c]} \psi$ to ensure that the robot arrives at co-observation exactly at the scheduled time,
    \item $\psi := (d(x,q)\leq r_1) \wedge (d(x,x')\leq r_2)$, where $d:\real{m}\times\real{m}\rightarrow \real{}$, as the Euclidean distance between locations, i.e. $d(x_1,x_2) = \norm{x_1-x_2}_2^2$. $x'$ being the location of the other robot for the co-observation and $r_1$, $r_2$ being the maximum offset allowed.
 \end{itemize}
Since the co-observation locations are fixed, we omitted the term $\norm{x-x'}\leq r_2$ for simplicity. Multiple co-observation requirements are combined through conjunctions, for example, $\phi = \phi_1 \wedge \phi_2 \wedge \phi_3$.

\subsubsection{Controller framework}\label{sec:control-design}
For implementation, the CLF-based controller is used to apply the trajectory tracking and online task target tracking requirements. 
On the other hand, the CBF-based security filter is used to apply security-related requirements. Using the CBF design method for STL tasks introduced in~\cite{lindemann2018control}, we design the CBF for simple temporal operators. Consider formula $F_{[a,b]} d(x,q)\leq r$, the candidate CBF is designed as:
\begin{equation}
    h(x,t) =\gamma_1(t) + \frac{r - d(x,q)}{v_{max}},
\end{equation}
where $v_{max}$ is the maximum speed for robot
\begin{equation}
    \gamma_1(t) = -\frac{a}{b}t+a.
\end{equation}

For formula $G_{[a,b]} d(x,q)\leq r$, the candidate CBF is designed as:
\begin{equation}
    h(x,t) =\gamma_2(t) + r - d(x,q),
\end{equation}
with 
\begin{equation}
    \gamma_2(t) = \mu e^{-\epsilon t} - \sigma,
\end{equation}
where $\epsilon$ is the decay rate to be designed, $r_1$  is a sufficiently small offset variable that ensures the function $\gamma_2(t)\leq 0$ for all $t\in [a,b]$, and $\mu= r_1 e^{\epsilon a}$. 

For conjunctions $\phi_1 \wedge \phi_2$, the candidate CBF is designed as follows:
\begin{equation}
    h(x,t) = -ln(e^{-h_1(x,t)} + e^{-h_2(x,t)}),
\end{equation}
where $h_1$ and $h_2$ are candidate CBF for $\phi_1$ and $\phi_2$ respectively.

To avoid conservatism and potential conflicts arising from multiple constraints, we focus solely on the CBF for the next scheduled co-observation and regroup co-observation. Once a co-observation is fulfilled, the corresponding CBF is deactivated.

We then propose to first solve the following local constrained optimization problem to get the optimal reference control law $u_{ref}$, for robot $i$, we have
\begin{equation}
\begin{aligned}
u_{ref}=&\argmin_{u_i\in U}&& u\transpose Q u \\
        & s.t.&& L_f V_t(x) + L_g V_t(x)u + \gamma(V(x))\\ 
        &    && \qquad \qquad \leq (1-\alpha_{iP}+\sum_j \alpha_{iP'_j})\cM \\
        &     && L_f V_{O_j}(x) + L_g V_{O_j}(x)u + \gamma(V(x))\\
        &     && \qquad \qquad \leq (1-\alpha_{iO_j})\cM, \quad \forall \{O_j\},
 \end{aligned}
\end{equation}
where $Q\in\real{n\times n}$ is a positive semi-definite weight matrix, and $\alpha_{i,P}$, $\alpha_{i,P'}$ and $\alpha_{i,O_j}$ are the assignment coefficient of robot $i$ for trajectory tasks, secondary trajectory task and online tasks respectively.

\begin{remark}\label{rmk:alpha-in-control}
    As noted in \cref{rmk:not-binary}, the assignment result $\alpha_{ij}$ is not always strictly binary. To incorporate $\alpha$ into constraint prioritization, we introduce a sufficiently large penalty constant $\cM$ and scale the constraint terms using $\alpha$. When $\alpha_{ij}$ is close to 1, the right-hand side of the constraints approaches $0$, giving higher priority to the corresponding CLF. Conversely, when $\alpha_{ij}$ approaches $0$, the constraint is relaxed, reflecting lower priority. A similar approach is applied in the CBF-based control scheme.
\end{remark}

Then we use CBF as a security filter:
\begin{equation}
    \begin{aligned}
        \min_{u\in U} & \quad (u-u_{ref})\transpose (u-u_{ref})&\\
        s.t. & \dot{h}_{O_j}(x,t) + \beta_T(h_{O_j}(x,t))\geq -(1-\alpha_{iO_j})\cM, \quad \forall \{O_j\} \\
            & \dot{h}_{P}(x,t) + \beta_T(h_{P}(x,t))\geq -(1-\alpha_{iP})\cM \\
            &\dot{h}_{O_j}(x,t) + \beta_T(h_{O_j}(x,t))\geq -(1-\alpha_{iP})\cM \\
            & L_f h_c(x) + L_g h_c(x)u + \beta_c(h_c(x))\geq 0, 
    \end{aligned}
\end{equation}
where $h_{O_j}$ is the candidate CBF for regroup co-observation associated with online task $O_j$, $h_{P}$ is the candidate CBF for the next scheduled co-observation, and $h_c$ represents CBF for collision avoidance and other time-invariant safety requirements~\cite{Li2017CBF}.

By using the two-step approach above (computation of the reference, followed by a security filter), the CLF constraints are relaxed while prioritizing the CBF constraints to maintain security. This avoids the potential conflict and infeasibility that may arise when combining CLF and CBF constraints into a single QP, where even penalized CLF constraints can interfere with strict safety enforcement. As discussed in~\cite{2021Reis}, such a combined formulation can also introduce undesired equilibrium which do not correspond to the minimum of the Lyapunov function. For example, when a robot gets blocked while executing an online task and a detour will cause it to miss the regroup co-observation, the robot prioritizes the regroup and abandons the online task.

\section{Simulation}\label{sec:simulation}
\begin{figure}
    \centering
    \includegraphics[width=0.8\linewidth,trim = 1cm 1cm 1cm 1cm, clip]{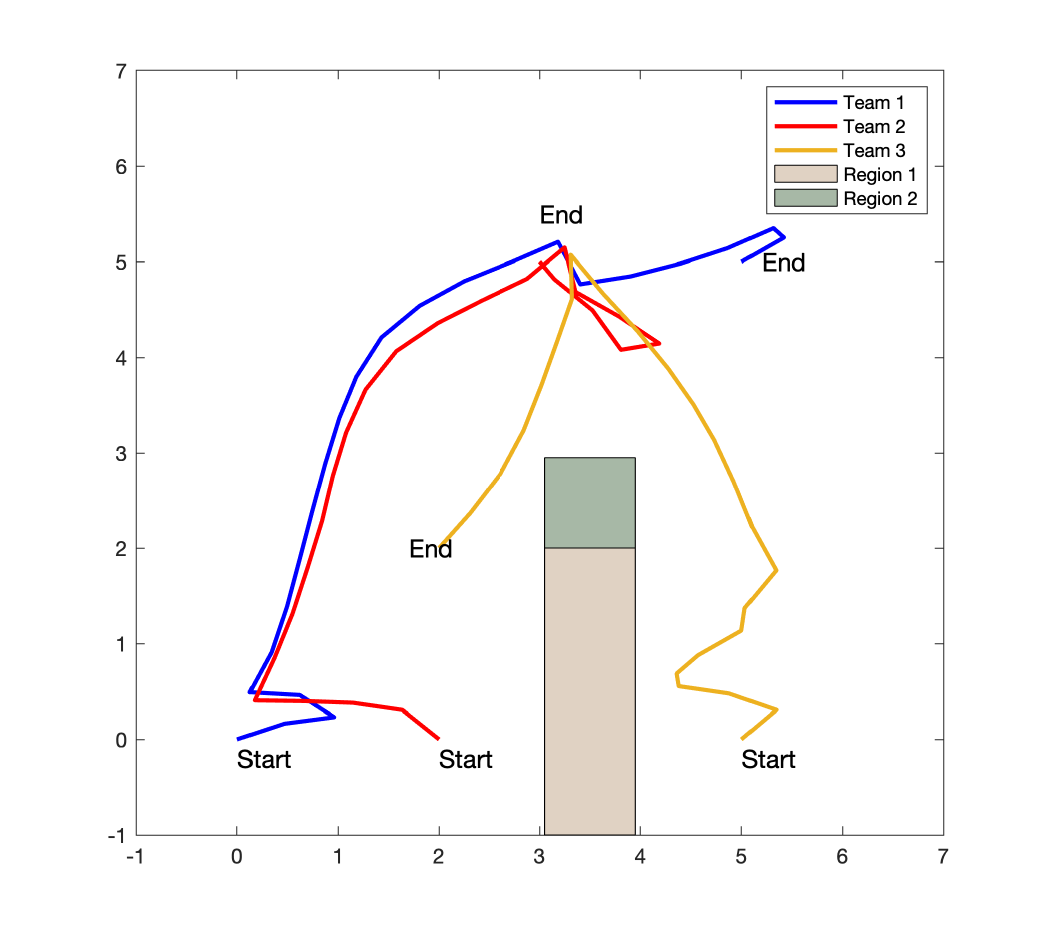}
    \caption{Secured 3-team reference path generated offline, region 1 is a forbidden region and region 2 is an obstacle.}\label{fig:PrePlanedPath}
\end{figure}

\begin{figure*}[h]
    \centering
    \subfloat[Time $t=12$. Agent 1 has been assigned to online task 1 while agent 2 and 3 follow the trajectory.\label{fig:t=12}]{\includegraphics[width=0.3\linewidth,trim = 1cm 1cm 1cm 1cm, clip]{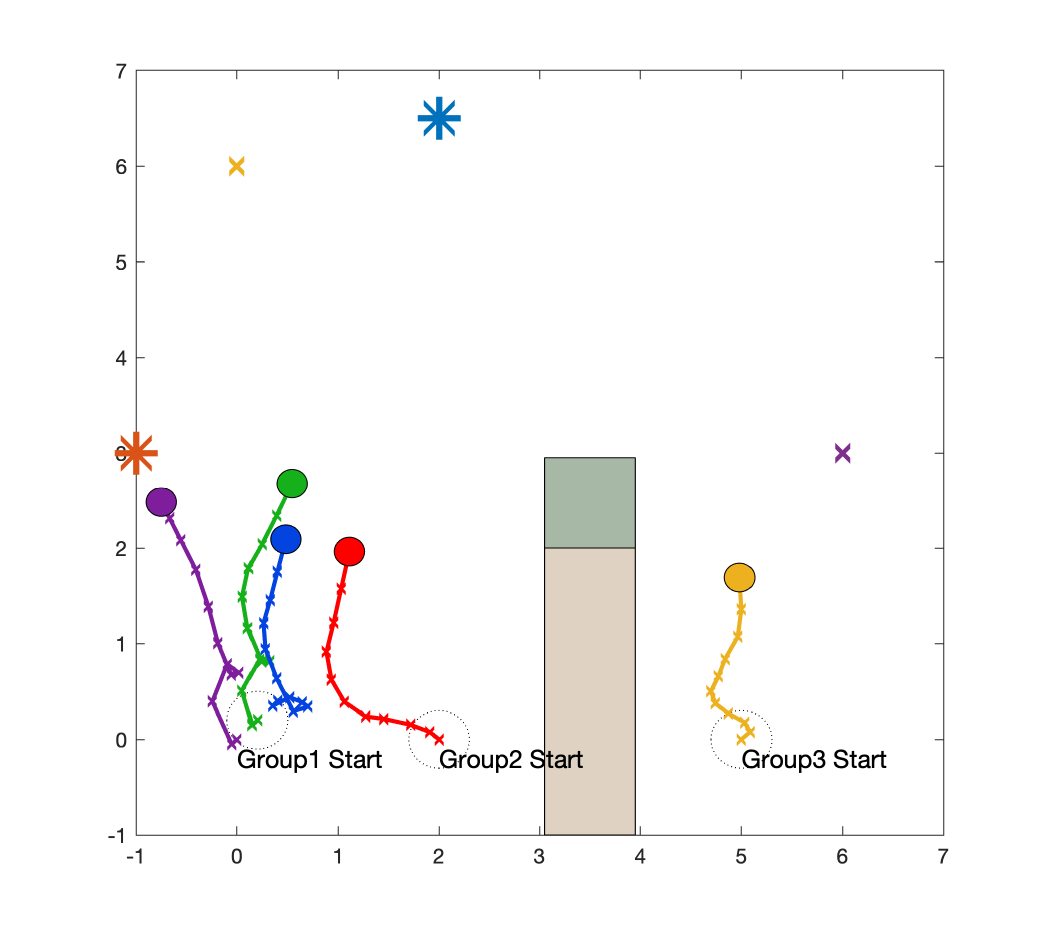}}
    \quad
    \subfloat[Time $t=20$. Agent 1 gets back to the trajectory and agent 2 has been assigned to online task 2.\label{fig:t=20}]{\includegraphics[width=0.3\linewidth,trim = 1cm 1cm 1cm 1cm, clip]{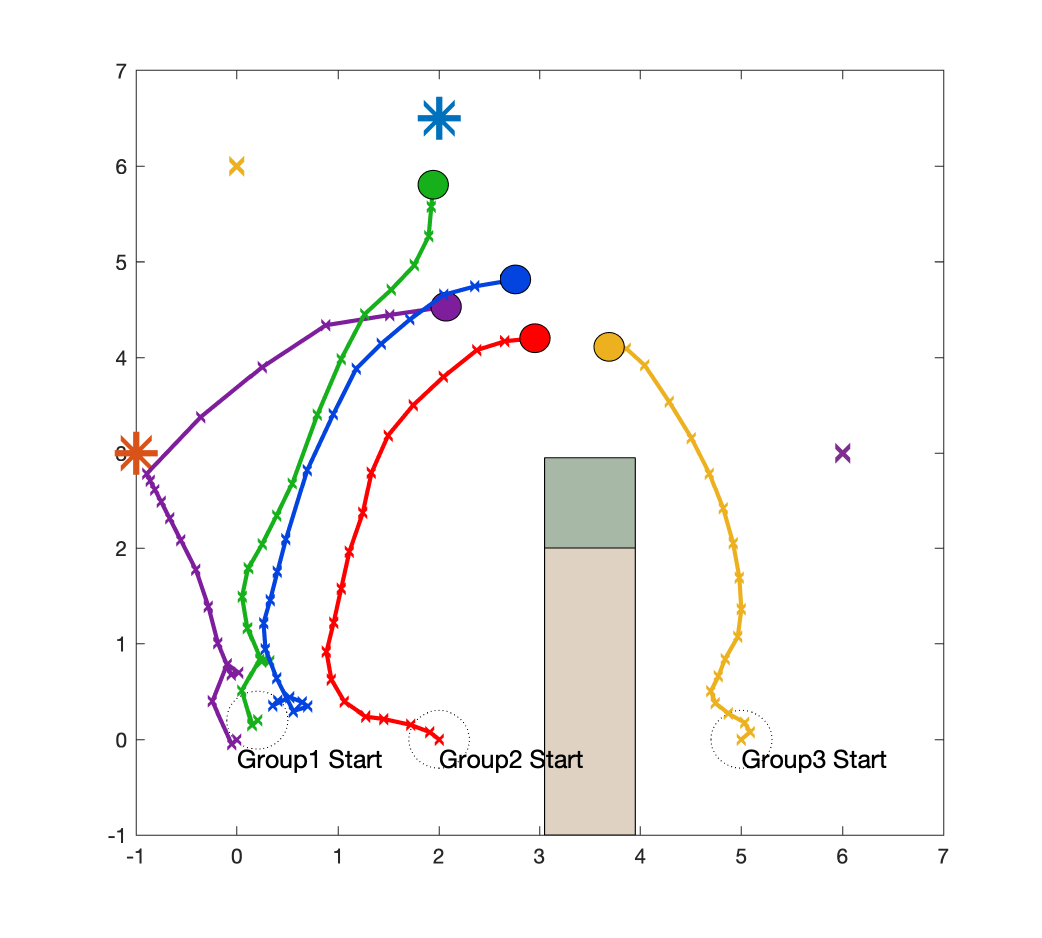}}
    \quad
    \subfloat[Runtime simulation with online tasks appear during the mission.\label{fig:OnlineFullpath}]{\includegraphics[width=0.3\linewidth,trim = 1cm 1cm 1cm 1cm, clip]{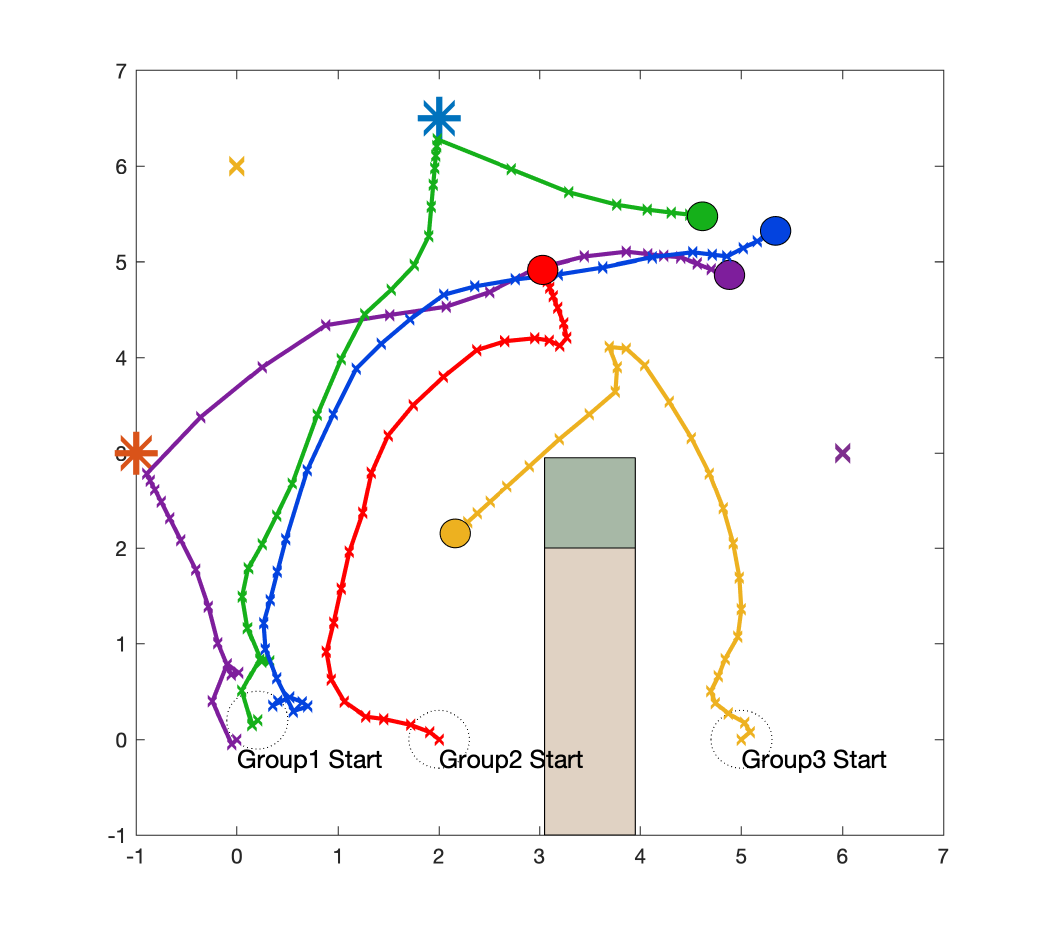}}
    \caption{ 3 robots form blue sub-team (following the blue trajectory) in \cref{fig:PrePlanedPath} to deal with online tasks. Sign $\ast$ indicating fulfilled security-verified online tasks, and sign $\times$ indicating unassigned online tasks.}
\end{figure*}
    
In this section, the proposed strategy is tested using a map exploration secured with a co-observation schedule for three sub-teams. We assume that the environment is a square region with the coordinate origin $(0,0)$ located in the bottom-left corner and the edge length of $8$ units. In the test area, \emph{zone 1} is an obstacle and \emph{zone 2} is a forbidden area.  The robot dynamics are modeled using a first-order single-integrator system, which is commonly used to validate high-level coordination strategies.  All robots are set to have a maximum velocity of $0.5m/s$ with a total task time of $20s$. The reference trajectory is generated using the method in \cite{yang2021multi} alongside a co-observation plan: team $1$ and $2$ meet at time $8$ and $14$, while team $2$ and $3$ meet at time $18$. The reference trajectory and the environment setup are shown in~\cref{fig:PrePlanedPath}. In real-time simulation, group $1$ has three robots in the sub-team while group $2$ and $3$ have one. We will focus on the task assignment performance and deadline setup for group 1. 

During the task period, two out of four \emph{online} tasks appeared have been fulfilled. When the mission first began, there were no \emph{online} tasks, and all agents were performing the trajectory task. When \emph{online} task $2$ appears and a secured regroup time was found at $t=24$ (as shown in~\cref{fig:t=12}), agent 1 was assigned with \emph{online} task 1. Meanwhile, agent 1 set up a deadline CBF to guarantee the return to reference trajectory before the \emph{regroup} time. 

When \emph{online} task 1 and 3 appear and are sent to the assignment algorithm later (as shown in~\cref{fig:t=20}), \emph{online} task 1 is assigned to agent 2, while task 3 is assigned to a shadow agent due to an insufficient number of agents (agent 3 was assigned to \emph{trajectory} task and agent 1 was still working on getting back to the reference trajectory). Online task 4 is never sent to the assignment process because it is always outside the secure reachability regions. 

During the entire task process, agent 3 strictly follows the trajectory and the co-observation schedule to vouch for the security of the team. The final result is shown in~\cref{fig:OnlineFullpath}. Both the task and co-observation had been satisfied, and when agents separate, they successfully return to the reference trajectory before the \emph{regroup} deadline.

\section{Conclusion and Future Work}\label{sec:Conclusion}
We proposed a distributed task assignment algorithm that dynamically allocates robots with different priorities. Using an inexact ADMM-based approach, the problem is decomposed into separable and non-separable subproblems, with the latter solved via projected gradient descent through local communication. This distributed formulation enables efficient coordination of security-related high-priority tasks and unplanned optional online tasks. We integrated this algorithm into a comprehensive framework that enables MRS to safely handle unplanned \emph{online} tasks, validated through real-time simulation. The proposed approach consists of a security analysis to determine whether an online task can be executed securely and, if so, the required time and location for the robot fulfilling it to return to the team. It also includes the distributed task assignment algorithm and an online controller that fulfills the assigned tasks while using a CBF-STL-based security filter to enforce security requirements.

In the future, we plan to investigate the effects of network topology changes during task execution. For instance, team composition may vary as robots join other sub-teams or move out of communication range while completing online tasks. We will further analyze the impact of these changes to ensure our approach remains robust and effective across different scenarios. Another future direction is the application of our algorithm to heterogeneous MRS. While robots can incorporate their specialization through local constraints, the current algorithm struggles with handling abandoned tasks when there are insufficient specialized robots to fulfill them. This limitation arises because the algorithm requires inserting shadow agents in advance. Addressing this challenge will be a key focus of our future work.

\appendix
\subsection{Proof of Distributed $\vz$-update}\label{app:proof}
First we consider the convergence of single iteration for \eqref{eq:z_update_distributed}. Consider a general ADMM formulation with Lagrangian \eqref{eq:Lagrangian}, z-update can be written as:
\begin{equation}
    z^{k+1} \leftarrow \argmin_z g(z) + (\rho/2) \norm{\alpha - z + u}_2^2
\end{equation}
which is equivalent to \eqref{eq:z_update}, and can be solved via one step of projected Gradient update of problem, which can be considered as:
\begin{equation}\label{eq:apgm}
    z^{k+1}[k'+1] \leftarrow \argmin_z g(z) + \frac{1}{\epsilon} \norm{z-(z^{k+1}[k'] - \epsilon \grad_{\vz_{:j}}\cL)}_2^2.
\end{equation}
When treat the quadratic part as a smooth term and $g(z)$ as a constraint via the indicator function, \eqref{eq:z_update_distributed} can be seen as the projected gradient step toward minimizing the Lagrangian.

According to \cite{ma2016alternating}, for the optimization problem \eqref{eq:square-problem} (let $f(\alpha) = \sum_i f_i( \pmb\alpha_i)$ for simplicity),
update step of variable $z$ with a step size of $\tau$ via:
\begin{multline}
    z^{k+1}:=\argmin_z g(z) +\\
     \frac{\rho}{2\tau}\norm{z-\left(z^k+\tau(x^{k+1}-z^k+u^k)\right)}_2^2,
\end{multline}
which renders that:
\begin{equation}\label{eq:partialKKT}
    0\in  \partial g(z^{k+1}) + \frac{\rho}{\tau} (z^{k+1}-z^{k}) + \rho \transpose (z^{k+1}-z^k-u^{k+1}).
\end{equation}
Using the fact that $\partial g$ is a monotonic operator, \eqref{eq:partialKKT} can be combined with the KKT condition $0\in \partial g(z^*)+ u^*$ to derive the condition:
\begin{multline}\label{eq:z_update_trand}
    (z^{k+1}-z^*)\big((\frac{1}{\tau}(z^k-z^{k+1})-\\
    (z^k-z^{k+1})-(u^{k+1}-u^*))\big)\geq 0.
\end{multline}
And since regular $\alpha$-update for ADMM is used, which renders:
\begin{equation}
    0 \in \partial f(\alpha^{k+1}) + \rho (u^{k+1}-z^k+z^{k+1}).
\end{equation}
Similarly, combined with KKT condition $0\in \partial f(\alpha^*) -  \alpha^*$ yields:
\begin{equation}\label{eq:x_update_trand}
    (\alpha^{k+1}-\alpha^{k})\transpose ( (u^{k+1}-u^*)+(z^{k+1}-z^k))\geq 0.
\end{equation}
Combine \eqref{eq:z_update_trand} and \eqref{eq:x_update_trand} considering the fact that $\alpha^* =z^* $ and $u^{k+1}= u^k+\alpha^{k+1}-z^{k+1}$, we get:
\begin{multline}
    \frac{1}{\tau}(z^{k+1}-z^*)\transpose(z^k-z^{k+1}) + (u^k-u^{k+1})\transpose (z^k-z^{k+1}) \\+ (u^{k+1}-u^*)\transpose(u^k-u^{k+1})\geq 0
\end{multline}
Considering having $y^k=\bmat{\frac{1}{\sqrt{\tau}} z^k \\ u^k}$, and $y^* = \bmat{\frac{1}{\sqrt{\tau}} z^* \\u^*}$, we can rewrite this inequality as:
\begin{multline}
 (y^{k+1}-y^{*})\transpose (y^k-y^{k+1})\\
 =  (y^{k}-y^{*})\transpose (y^k-y^{k+1}) - \norm{y^k-y^{k+1}}_2^2\\
 \geq -(u^k-u^{k+1})\transpose (z^k-z^{k+1})
\end{multline}
which implies
\begin{multline}\label{eq:update2optimal}
    \norm{y^k-y^*}^2_2-\norm{y^{k+1}-y^*}^2_2 \\
    = 2(y^{k}-y^{*})\transpose (y^k-y^{k+1}) - \norm{y^k-y^{k+1}}^2_2 \\
    \geq  \norm{y^k-y^{k+1}}^2_2-2(u^k-u^{k+1})\transpose (z^k-z^{k+1})
\end{multline}
Set $\xi = \frac{1}{2}+\frac{\tau}{2}$, and $\tau<\xi<1$, using the Cauchy-Schwartz inequality, we can get
\begin{multline}
-2(u^k-u^{k+1})\transpose (z^k-z^{k+1}) \\\geq -\xi \norm{u^k-u^{k+1}}^2_2 - \frac{1}{\xi}\norm{z^k-z^{k+1}}^2_2
\end{multline}
making \eqref{eq:update2optimal}
\begin{multline}\label{eq:convergence}
    \norm{y^k-y^*}^2_2-\norm{y^{k+1}-y^*}^2_2 \\
    \geq  (1-\xi) \norm{u^k-u^{k+1}}^2_2 + (\frac{1}{\tau}-\frac{1}{\xi})\norm{z^k-z^{k+1}}^2_2\\
    \geq \eta \norm{y^k-y^{k+1}}
\end{multline}
where $\eta = \min\{ (1-\xi), (\frac{1}{\tau}-\frac{1}{\xi})\}$.
From \eqref{eq:convergence}, it can be implied that $\norm{y^k-y^{k+1}}^2_2\rightarrow 0$, and $\norm{y^k-y^*}^2_2$ is monotonically non-increasing and thus converges. Thus, if the proximal gradient step $\epsilon < 1$, $z^k$ provided by \eqref{eq:z_update_distributed_agent} will converge to optimal solution $z^*$.
{\small
\bibliographystyle{IEEEtran}
\bibliography{ADMM_planning,ACC,task_assignment}
}
\begin{IEEEbiography}[{\includegraphics[width=1in,height=1.25in,clip,keepaspectratio]{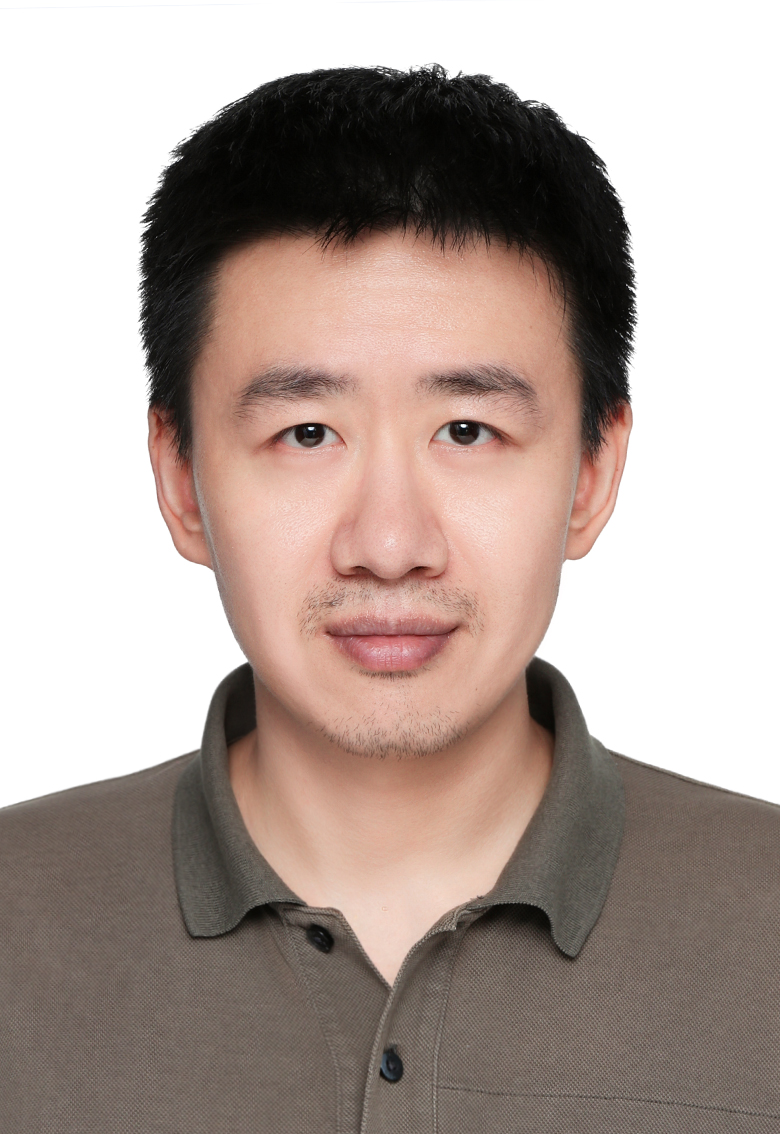}}]{Ziqi Yang} received the B.S. degree in Automation from Harbin Institute of Technology, Harbin, China in 2015, the M.Eng. degree in Electrical and Computer Engineering from Cornell University, Ithaca, NY, USA in 2017, and the Ph.D. degree in Systems Engineering from Boston University, Brookline, MA, USA in 2023. He is currently a Research Fellow at the School of Electrical and Electronic Engineering at Nanyang Technological University.
His current research interests include nonlinear control, planning and optimization, with an emphasis on system safety and security.
    \end{IEEEbiography}
    \begin{IEEEbiography}[{\includegraphics[width=1in,height=1.25in,clip,keepaspectratio]{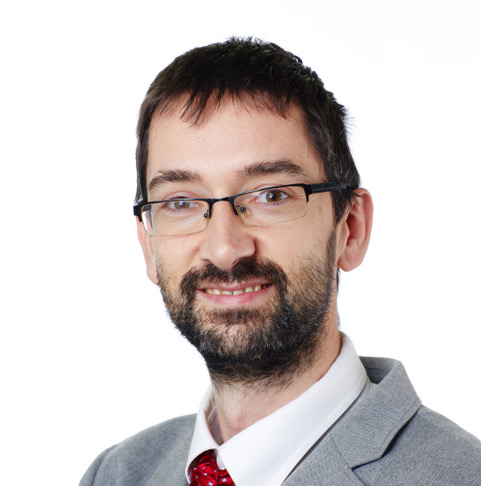}}]{Roberto Tron}
        received the B.Sc. degree in 2004
        and the M.Sc. degree (highest honors) in 2007 both in telecommunication engineering from the Politecnico di Torino, Turin, Italy, the Diplome d’Engenieur from the Eurecom Institute, Biot, France, and the DEA degree from the Université de Nice Sophia-Antipolis, Nice, France in 2006, and the Ph.D. degree in electrical and computer engineering from The Johns Hopkins University, Baltimore, MD, USA, in 2012. He was a Postdoctoral Researcher with the GRASP Lab, University of Pennsylvania, Philadelphia, PA, USA, until 2015. He is an Associate Professor of Mechanical Engineering and System Engineering with Boston University, Boston, MA, USA. His research interests include the intersection of automatic control, robotics, and computer vision,  with a particular interest in applications of Riemannian geometry and linear programming in problems involving distributed teams of agents, or geometrical and spatio-temporal constraints. Dr. Tron was recognized at the IEEE Conference on Decision and Control with the “General Chair’s Interactive Presentation Recognition Award”, in 2009, the “Best Student Paper Runner-up” in 2011, and the “Best Student Paper Award” in 2012.
    \end{IEEEbiography}

\end{document}

%% file: preamble/common.tex
\input{preamble/commonNoTikz}
\input{preamble/graphicsTikz}
\input{preamble/robotics}


%% file: preamble/commonNoTikz.tex
\input{preamble/fixes}
\input{preamble/graphics}
\input{preamble/pagination}
\input{preamble/utilities}
\input{preamble/math}
\input{preamble/operators}

\input{preamble/notation}

\input{preamble/units}

\input{preamble/markupAndCommenting}


%% file: preamble/fixes.tex
\usepackage{fix-cm}
\usepackage{etex}

\usepackage{dblfloatfix}

\usepackage{nag}


%% file: preamble/graphics.tex
\makeatletter
\@ifpackageloaded{xcolor}{}{%
\usepackage[table,x11names,dvipsnames,svgnames]{xcolor}%
}
\makeatother

\usepackage{colortbl}

\usepackage{graphicx}
\usepackage{wrapfig}

\input{preamble/graphicsColors}


%% file: preamble/graphicsColors.tex
\definecolorset{RGB}{lyft}{}{Red,194,39,36;Sunset,202,53,33;Orange,205,68,20;Amber,200,117,42;Yellow,242,169,52;Citron,186,188,44;Lime,112,159,33;Green,56,139,31;Mint,45,118,56;Teal,52,133,135;Cyan,60,132,202;Blue,55,94,248;Indigo,64,13,247;Purple,115,42,248;Pink,176,25,145;Rose,176,32,75}

%% file: preamble/pagination.tex
\usepackage{cite}

\usepackage{microtype}

\usepackage{array}
\usepackage{multirow}
\usepackage{booktabs}
\usepackage{makecell} 

\ifcsname labelindent\endcsname

\fi
\usepackage[inline]{enumitem}

\usepackage{subfig}

\newenvironment{lenumerate}[2][]
{\begin{enumerate}[label=(#2\arabic*),leftmargin=0.2in,itemindent=0.15in,#1]}
{\end{enumerate}}

\setlist*[enumerate,1]{label={\itshape\arabic*)}}

\makeatletter
\newcommand{\paragraphswithstop}{%
\let\copyparagraph\paragraph%
\renewcommand\paragraph[1]{\copyparagraph{##1.}}%
}
\makeatother

\usepackage[framemethod=tikz]{mdframed}


%% file: preamble/utilities.tex
\usepackage{suffix}

\usepackage{environ}

\makeatletter
\newsavebox{\boxifnotempty}
\newcommand{\displayifnotempty}[3]{\sbox\boxifnotempty{#2}\setbox0=\hbox{\usebox{\boxifnotempty}\unskip}%
\ifdim\wd0=0pt
\else
 #1\usebox{\boxifnotempty}#3%
\fi%
}
\newcommand{\ifempty}[2]{\setbox0=\hbox{#1\unskip}%
\ifdim\wd0=0pt%
 #2%
\fi%
}
\newcommand{\ifnotempty}[2]{\setbox0=\hbox{#1\unskip}%
\ifdim\wd0>0pt%
 #2%
\fi%
}
\makeatother

\usepackage{algpseudocode}
\usepackage{algorithm}

\usepackage{scrlfile}

\makeatletter
\newcommand*\newstoreddef[1]{
  \BeforeClosingMainAux{%
    \immediate\write\@auxout{%
      \string\restoredef{#1}{\csname #1\endcsname}%
    }%
  }%
}
\newcommand*{\restoredef}[2]{
  \expandafter\gdef\csname stored@#1\endcsname{#2}%
}
\newcommand*{\storeddef}[1]{
  \@ifundefined{stored@#1}{0}{\csname stored@#1\endcsname}%
}
\makeatother

\usepackage{pageslts}
\NewEnviron{tee}{\BODY\typeout{Marker Tee [start] ^^J \BODY ^^JMaker Tee [end]}}


%% file: preamble/operators.tex
\newcommand{\real}[1]{\mathbb{R}^{#1}{}}

\newcommand{\bmat}[1]{\begin{bmatrix}#1\end{bmatrix}}

\newcommand{\transpose}{^\mathrm{T}}

\newcommand{\inverse}{^{-1}}

\DeclarePairedDelimiter{\abs}{\lvert}{\rvert}

\DeclarePairedDelimiter{\norm}{\lVert}{\rVert}

\newcommand{\vct}[1]{\mathbf{#1}}

\DeclareMathOperator{\diag}{diag}

\DeclareMathOperator*{\argmin}{\arg\!\min}

\DeclareMathOperator{\trace}{tr}

\DeclareMathOperator{\stack}{stack}

\DeclareMathOperator{\grad}{{grad}}

\newcommand{\union}{\cup}

\newcommand{\subjectto}{\textrm{subject to }}


%% file: preamble/notation.tex
\providecommand{\vq}{\vct{q}}

\providecommand{\vu}{\vct{u}}

\providecommand{\vv}{\vct{v}}

\providecommand{\vw}{\vct{w}}

\providecommand{\vz}{\vct{z}}

\providecommand{\cA}{\mathcal{A}}

\providecommand{\cC}{\mathcal{C}}

\providecommand{\cE}{\mathcal{E}}
\providecommand{\cF}{\mathcal{F}}

\providecommand{\cI}{\mathcal{I}}

\providecommand{\cK}{\mathcal{K}}
\providecommand{\cL}{\mathcal{L}}
\providecommand{\cM}{\mathcal{M}}

\providecommand{\cS}{\mathcal{S}}
\providecommand{\cT}{\mathcal{T}}



%% file: preamble/units.tex
\usepackage{units}


%% file: preamble/markupAndCommenting.tex
\newcommand{\newcolorlabel}[2]{%
  \expandafter\newcommand\csname #1\endcsname[1]{%
    \colorbox{#2}{\color{white}\textsf{\textbf{##1}}}}%
}

\newcommand{\newcommenter}[2]{%
  \expandafter\newcommand\csname #1\endcsname[1]{%
    \fcolorbox{#2}{#2}{\color{white}\textsf{\textbf{#1}}}
    {\color{#2}##1}}%
  \expandafter\newcommand\csname at#1\endcsname{%
    \fcolorbox{#2}{#2}{\color{white}\textsf{\textbf{@#1}}}
    {\color{#2}}}%
  \expandafter\newcommand\csname #1hl\endcsname[2]{%
    \colorbox{#2}{\color{white}\textsf{\textbf{#1}}}\sethlcolor{Azure2}\hl{##2}~%
    \expandafter\ifx\csname commentarrow\endcsname\relax$\leftarrow$\else \commentarrow[#2]\fi~%
    {\color{#2}##1}}%
  \expandafter\newcommand\csname #1st\endcsname[2]{%
    \colorbox{#2}{\color{white}\textsf{\textbf{#1}}}\sout{##2}~%
    \expandafter\ifx\csname commentarrow\endcsname\relax$\leftarrow$\else \commentarrow[#2]\fi~%
    {\color{#2}##1}}%
}
\newcommenter{TODO}{DodgerBlue1}
\newcommenter{rtron}{Green3}
\newcommenter{zyang}{blue}

\usepackage{comment}

\usepackage{pdfcomment}

\usepackage{soul}

\usepackage[normalem]{ulem}

\usepackage{csquotes}


%% file: preamble/graphicsTikz.tex
\usepackage{tikz}
\usetikzlibrary{calc}
\usetikzlibrary{matrix}
\usetikzlibrary{chains}
\usetikzlibrary{shapes.geometric}
\usetikzlibrary{arrows.meta}
\usetikzlibrary{decorations.pathreplacing}
\usetikzlibrary{backgrounds}

\tikzset{
  dim above/.style={to path={\pgfextra{
        \pgfinterruptpath
        \draw[>=latex,|->|] let
        \p1=($(\tikztostart)!1.5em!90:(\tikztotarget)$),
        \p2=($(\tikztotarget)!1.5em!-90:(\tikztostart)$)
        in(\p1) -- (\p2) node[pos=.5,sloped,above]{#1};
        \endpgfinterruptpath
      }
    }
  },
  dim double above/.style={to path={\pgfextra{
        \pgfinterruptpath
        \draw[>=latex,|->|] let
        \p1=($(\tikztostart)!3em!90:(\tikztotarget)$),
        \p2=($(\tikztotarget)!3em!-90:(\tikztostart)$)
        in(\p1) -- (\p2) node[pos=.5,sloped,above]{#1};
        \endpgfinterruptpath
      }
    }
  },
  dim below/.style={to path={\pgfextra{
        \pgfinterruptpath
        \draw[>=latex,|->|] let 
        \p1=($(\tikztostart)!-1em!-90:(\tikztotarget)$),
        \p2=($(\tikztotarget)!-1em!90:(\tikztostart)$)
        in (\p1) -- (\p2) node[pos=.5,sloped,below]{#1};
        \endpgfinterruptpath
      }
    }
  },
}

\tikzset{
    right angle quadrant/.code={
        \pgfmathsetmacro\quadranta{{1,1,-1,-1}[#1-1]}     
        \pgfmathsetmacro\quadrantb{{1,-1,-1,1}[#1-1]}},
    right angle quadrant=1, 
    right angle length/.code={\def\rightanglelength{#1}},   
    right angle length=2ex, 
    right angle symbol/.style n args={3}{
        insert path={
            let \p0 = ($(#1)!(#3)!(#2)$) in     
                let \p1 = ($(\p0)!\quadranta*\rightanglelength!(#3)$), 
                \p2 = ($(\p0)!\quadrantb*\rightanglelength!(#2)$) in 
                let \p3 = ($(\p1)+(\p2)-(\p0)$) in  
            (\p1) -- (\p3) -- (\p2)
        }
    }
}

\newcommand{\pgfextractangle}[3]{%
    \pgfmathanglebetweenpoints{\pgfpointanchor{#2}{center}}
                              {\pgfpointanchor{#3}{center}}
    \global\let#1\pgfmathresult  
}

\usetikzlibrary{shapes.arrows}
\newcommand{\commentarrow}[1][Azure4]{\tikz[baseline=-3pt]{\node[shape border uses incircle, fill=#1,rotate=180,single arrow, inner sep=1pt, minimum size=6pt, single arrow head extend=2pt]{};}}

%% file: preamble/robotics.tex
\tikzset{ax/.style={-latex,line width=2pt}}

\tikzset{camera/.style={fill=Sienna1,fill opacity=0.5},%
image plane/.style={draw=RoyalBlue3,line width=2pt}}
